\documentclass[%
 reprint,
superscriptaddress,
nobibnotes,
 amsmath,amssymb,
 aps,
 prx,
]{revtex4-2}

\usepackage{comment}
\usepackage{graphicx}
\usepackage{dcolumn}
\usepackage{bm}

\usepackage{appendix}

\usepackage{graphicx}
\usepackage{dcolumn,amsthm}
\usepackage{bm}

\usepackage{txfonts}
\usepackage{dsfont}
\usepackage[english]{babel}

\usepackage[colorlinks=true,linkcolor=blue,citecolor=blue,urlcolor=blue]{hyperref}

\usepackage{enumitem}

\newcommand{\abs}[1]{\left\lvert#1\right\rvert}

\providecommand{\tr}{{\rm Tr}}
\renewcommand{\phi}{\varphi}
\newcommand{\bra}[1]{\left\langle #1\right\rvert}
\newcommand{\ket}[1]{\left\lvert #1\right\rangle}

\newtheorem{theorem}{Theorem}
\newtheorem{lemma}{Lemma}
\newtheorem{definition}{Definition}

\begin{document}


\title{Resources for bosonic quantum computational advantage}

\author{Ulysse Chabaud}
\email{ulysse.chabaud@inria.fr}
\affiliation{Institute for Quantum Information and Matter, California Institute of Technology, 1200 E California Blvd, Pasadena, CA 91125, USA}
\affiliation{DIENS, \'Ecole Normale Sup\'erieure, PSL University, CNRS, INRIA, 45 rue d'Ulm, Paris 75005, France}
\author{Mattia Walschaers}
\email{mattia.walschaers@lkb.upmc.fr}
\affiliation{Laboratoire Kastler Brossel, Sorbonne Universit\'{e}, CNRS, ENS-Universit\'{e} PSL,  Coll\`{e}ge de France, 4 place Jussieu, F-75252 Paris, France}




\begin{abstract}
Quantum computers promise to dramatically outperform their classical counterparts. However, the non-classical resources enabling such computational advantages are challenging to pinpoint, as it is not a single resource but the subtle interplay of many that can be held responsible for these potential advantages. In this work, we show that every bosonic quantum computation can be recast into a continuous-variable sampling computation where all computational resources are contained in the input state. Using this reduction, we derive a general classical algorithm for the strong simulation of bosonic computations, whose complexity scales with the non-Gaussian stellar rank of both the input state and the measurement setup. We further study the conditions for an efficient classical simulation of the associated continuous-variable sampling computations and identify an operational notion of non-Gaussian entanglement based on the lack of passive separability, thus clarifying the interplay of bosonic quantum computational resources such as squeezing, non-Gaussianity and entanglement.
\end{abstract}

\maketitle

\textit{Introduction.---}Ever since the earliest quantum algorithms \cite{doi:10.1098/rspa.1992.0167,doi:10.1137/S0097539796298637,doi:10.1137/S0036144598347011}, it has been clear that quantum computing holds the potential of reaching exponential speed-ups as compared to classical computers---be it for very specific problems. The computational advantage \footnote{There is some ambiguity in literature surrounding the terms ``quantum supremacy'', ``quantum advantage'' and ``quantum speedup''. In this work, we use the term ``quantum computational advantage'' to refer to any quantum computational protocol that cannot be efficiently executed by a classical machine.} of quantum computers was more rigorously established by connecting the classical simulation of certain quantum sampling problems to the collapse of the polynomial hierarchy of complexity classes \cite{doi:10.1098/rspa.2010.0301,10.1145/1993636.1993682}. Boson Sampling, in particular, has drawn the attention of a part of the physics community, because the protocol is naturally implemented with indistinguishable photons and linear optics. These sampling problems also lie at the basis of the random circuit sampling protocol \cite{RCS}, which would lead to the first experimental claim of a quantum computational advantage \cite{QuantumSupremacy}. However, in a game of constantly shifting goal posts, this claim has already been challenged \cite{pan2021solving}. 

At the same time, the development of building blocks for potential quantum computing hardware has drastically accelerated during the last decade. Even though platforms such as superconducting circuits and trapped ions have booked great successes, the present work mainly focuses on optical implementations. The Knill--Laflamme--Milburn scheme \cite{KLM} provided the first proposal for a universal photonic quantum computer, which to this day remains extremely challenging to implement. Even though Boson Sampling \cite{10.1145/1993636.1993682} renewed the interest in photonic quantum computing, generating, controlling, and detecting sufficiently many indistinguishable photons is still very challenging. 

To circumvent the difficulties of dealing with single photons and conserve the advantages that optics can provide for quantum information processing, such as intrinsic resilience against decoherence, several research groups have explored continuous-variable (CV) quantum optics as an alternative. Rather than detecting photons, this approach encodes information in the quadratures of the electromagnetic field, which can be detected through either homodyne or double homodyne (sometimes called heterodyne) measurements~\cite{Leonhardt-essential}. Equipped with its own framework for quantum computing in infinite-dimensional Hilbert spaces \cite{gu_quantum_2009}, the CV approach has the advantage of deterministic generation of large entangled states, over millions of subsystems \cite{Su:12,PhysRevLett.112.120505,Asavanant:2019aa,Larsen:2019aa,cai-2017}. By now, CV quantum optics is considered a promising platform for quantum computing \cite{Bourassa2021blueprintscalable}.
Several sampling problems have also been translated to an infinite-dimensional context \cite{PhysRevLett.113.100502,PhysRevLett.118.070503,PhysRevLett.119.170501,PhysRevA.96.062307,PhysRevA.96.032326}. Among these proposals, Gaussian Boson Sampling in particular attracted much attention, which led ultimately to experimental realisations beyond the reach of classical computers \cite{doi:10.1126/science.abe8770,PhysRevLett.127.180502,madsen2022quantum}. 

From a complexity-theoretic point of view, it is well understood why some of these specific sampling problems cannot be efficiently simulated by a classical computer \cite{doi:10.1126/sciadv.abi7894}.
From a physical point of view, several groups have explored the required resources for reaching a quantum computational advantage. Such endeavours typically aim to identify a physical property without which a setup can be efficiently simulated classically. Phase-space descriptions of quantum computations, such as the Wigner function~\cite{Wigner1932,gross2006hudson}, are particularly useful in that respect. For example, it has been shown that negativity of the Wigner function is one of such necessary resources \cite{PhysRevLett.109.230503,Veitch_2013}, albeit not sufficient~\cite{garcia2020efficient}. More recently, it became clear that squeezing and entanglement also play an important role in the hardness of some sampling problems, but only when combined in the right way \cite{PhysRevResearch.3.033018,chabaud2021holomorphic}. In Gaussian Boson Sampling, for example, the state at hand is an entangled Gaussian state, which can be described using a positive Wigner function, while negativity of the Wigner function is provided by the non-Gaussian photon detectors. This potential resourcefulness of the measurements is one reason why sampling problems are complicated to analyse.

In this work, we address this problem by introducing a new paradigm for studying resources for bosonic computations. Our contribution is three-fold:
Firstly, we show that every bosonic sampling computation has a dual CV sampling setup where the measurement is performed using double homodyne detection, which can be understood as quasi-classical and thus non-resourceful. This means that, in this dual sampling setup, all computational resources are ingrained in the measured state. 
Secondly, using this construction, we obtain a classical algorithm for strongly simulating bosonic computations, whose complexity scales with the stellar rank, a discrete non-Gaussian measure~\cite{PhysRevLett.124.063605}, of both the input state and the measurement setup of the original computation. Our algorithm is a generalisation of that of~\cite{chabaud2020classical}---which applies only to a restricted set of bosonic computations---to essentially any bosonic computation. Our result thus establishes the stellar rank as a necessary non-Gaussian resource for reaching a quantum computational advantage with bosonic information processing.
Thirdly, we further show that the associated CV sampling setup can also be efficiently simulated classically whenever its corresponding input state is passively separable. We explain that states that are not passively separable possess non-Gaussian entanglement, thus showing that this type of entanglement is necessary for reaching a quantum computational advantage.
Our results allow us to clarify the role played by different non-classical resources in enabling quantum computational advantage, which we illustrate with the example of Boson Sampling.

\textit{Sampling tasks.---}Our starting point is that of a general sampling setup, where a quantum state $\hat\rho$ over $m$ subsystems, or modes, is measured by a series of $m$ local detectors. We assume that the $k^{th}$ detector measures an observable $\hat Y_k$ with a spectral decomposition $\hat Y_k = \smash{\int_{{\cal Y}_k}y\hat P_{k;y_k} dy}$, where ${\cal Y}_k$ is the spectrum of $\hat Y_k$. Here, we limit ourselves to projective measurements, but our results can be extended to more general positive operator-valued measures through Naimark's dilation theorem. 

In a sampling setup, our goal is to sample detector outcomes with respect to the probability distribution given by the Born rule: $P(y_1,\dots, y_m \lvert \hat \rho):= \tr \left[\hat \rho \bigotimes_{k=1}^m\hat P_{k;y_k}\right]$.
For simplicity, we can assume that the projectors are rank-one, such that $\hat P_{k;y_k} = \ket{y_k}\!\bra{y_k}$. The measurement can thus be resourceful if $\ket{y_k}$ has a negative Wigner function or if it contains squeezing. A priori, the state $\hat\rho$ can be any multimode mixed state, but in a typical sampling setup it would be generated by applying a series of few-mode gates to a set of single-mode input states.

\textit{Stellar hierarchy.---}Hereafter, we describe bosonic states using the stellar hierarchy~\cite{PhysRevLett.124.063605} (see the Appendix \ref{app:stellar} for a concise review). This formalism associates to each $m$-mode pure state $\ket{\bm\psi}=\sum_{\bm n\ge\bm0}\psi_{\bm n}\ket{\bm n}$ its stellar (or Bargmann) function $\smash{F^\star_{\bm\psi}(\bm z)=\sum_{\bm n\ge\bm0}\frac{\psi_{\bm n}}{\sqrt{\bm n!}}\bm z^{\bm n}}$, for all $\bm z\in\mathbb C^m$, and classifies bosonic states according to their stellar rank: pure states of finite stellar rank $r^\star$ are those states whose stellar function is of the form $F^\star(\bm z)=P(\bm z)G(\bm z)$, where $P$ is a multivariate polynomial of degree $r^\star$ and $G$ is a multivariate Gaussian. Such states can be decomposed as $\hat G\ket{\bm C}$, where $\hat G$ is a Gaussian unitary and $\ket{\bm C}$ is a core state, i.e.\ a finite superposition of Fock states. The number of nonzero coefficients of $\ket{\bm C}$ is called the core state support size.
For mixed states, the stellar rank is defined by a convex roof construction: $\smash{r^\star(\hat\rho)=\inf_{p_i,\bm\psi_i}\sup r^\star(\bm\psi_i)}$, where the infinimum is over the decompositions $\hat\rho=\smash{\sum_ip_i\ket{\bm\psi_i}\!\bra{\bm\psi_i}}$. The stellar rank is a faithful and operational non-Gaussian measure~\cite{chabaud2021holomorphic}, as it is invariant under Gaussian unitaries, non-increasing under Gaussian maps, and it lower bounds the minimal number of non-Gaussian operations (such as photon additions or photon subtractions) necessary to prepare a bosonic state from the vacuum, together with Gaussian unitary operations. Moreover, any state can be approximated arbitrarily well in trace distance by states of finite stellar rank, and an optimal approximating state of a given stellar rank can be found efficiently~\cite{chabaud2020certification}.

To establish the duality between sampling an outcome from the distribution $P(y_1,\dots, y_m)$ and double homodyne sampling, we must analyse the pure states $\ket{y_k}$. It is convenient to use the stellar hierarchy to describe them: we can represent any single-mode state $\ket{y_k}$ of finite stellar rank as \cite{PhysRevLett.124.063605}
\begin{equation}\label{eq:yk}
    \ket{y_k} = \frac1{\sqrt{{\cal N}_k}} \left[\prod_{j=1}^{r^{\star}(y_k)} \hat D(\beta_{k;j}) \hat a^{\dag}_k \hat D^{\dag}(\beta_{k;j})\right]\ket{G_k},
\end{equation}
where $\smash{r^{\star}(y_k)}\in\mathbb N$ denotes the stellar rank of the state $\ket{y_k}$, $\ket{G_k}$ is a Gaussian state, $\hat D(\beta_{k;j})$ is a displacement operator that acts on mode $k$ with $\beta_{k;j}\in\mathbb C$, $\smash{\hat a^{\dag}_k}$ is the creation operator in mode $k$, and ${\cal N}_k$ is a normalisation factor. In this case, we can interpret $\ket{y_k}$ as an $r^{\star}(y_k)$-photon-added Gaussian state (when $\smash{r^{\star}(y_k)}=0$, the empty product is the identity operator by convention). Furthermore, since we can approximate any state $\ket{y_k}$ by a finite-rank state to arbitrary precision in trace distance, we assume that all $\ket{y_k}$ have a---possibly high---finite stellar rank.

The single-mode Gaussian states $\ket{G_k}$ can always be obtained from the vacuum with squeezing and displacement operations. This allows us to write $\ket{G_k} = \hat S_k \ket{\alpha_k} $, where $\hat S_k$ is a suitably chosen squeezing operation and $\ket{\alpha_k}=\hat D(\alpha_k)\ket0_{vac}$ is a coherent state. Combining this with Eq.~\eqref{eq:yk} we can now recast
\begin{equation}\label{eq:sampling}
    P(y_1,\dots, y_m \lvert \hat \rho) = \frac1{\cal N} \tr \left[ \hat{\bm S}^{\dag} \hat \rho^- \hat{\bm S} \bigotimes_{k=1}^m\ket{\alpha_k}\!\bra{\alpha_k} \right],
\end{equation}
where $\hat{\bm S}:=\smash{\bigotimes_k}\hat S_k$ and $\hat \rho^-$ is a non-normalised photon-subtracted state, given by $\smash{\hat \rho^- := \hat{\bm A} \hat \rho \hat{\bm A}^{\dag}}$, with a photon-subtraction operator $\smash{\hat{\bm A}:= \bigotimes_{k=1}^m\prod_{n=1}^{r^{\star}(y_k)} \hat D(\beta_{k;n}) \hat a_k \hat D^{\dag}(\beta_{k;n})}$.
The normalisation factor ${\cal N}$ in Eq.~\eqref{eq:sampling} is directly related to the detectors we use, thus we assume it to be known a priori.

\textit{Coherent state samplers.---}Double homodyne measurement corresponds to a (subnormalised) projection onto coherent states~\cite{Leonhardt-essential}. Hence, the expression in Eq.~\eqref{eq:sampling} shows that sampling measurement outcomes $y_1,\dots, y_m$ can always be connected to performing double homodyne measurements on a state that is obtained by squeezing and subtracting photons from the initial state $\hat \rho$. The implementation of photon subtraction generally requires measurements on auxiliary modes. The most common implementation involves a photon-counting measurement~\cite{lvovsky2020production}, but this is not compatible with our aim of not having any resources at the level of the measurement, since these measurements are represented by negative Wigner functions. Thus, we introduce a more unusual construction inspired by sum-frequency generation \cite{boyd2020nonlinear}. 

To subtract a photon in a mode $k$ from a state $\hat\rho$, we attach an auxiliary mode to our system, containing exactly one photon. This state is injected in a very weak two-mode squeezer, given by a unitary $\hat U(\xi) = \exp[i\xi(\hat a_k^{\dag}\hat a_{\rm aux}^{\dag} + \hat a_k\hat a_{\rm aux})]$ (acting as identity on all except the $k^{th}$ and the auxiliary modes). After having applied $\hat U(\xi)$, we project the auxiliary mode on the vacuum state to find $\tr_{\rm aux} \left\{\hat U(\xi) [\hat \rho \otimes \ket{1}\!\bra{1}] \hat U^\dag(\xi) [\hat{\mathds{1}} \otimes \ket{0}\!\bra{0}]\right\}\approx \xi^2\hat a_k\hat\rho\hat a_k^{\dag}$, where the approximation becomes exact when the approximation parameter $\xi$ goes to $0$ (see Appendix \ref{app:photsub}).
Replacing each photon subtraction in Eq.~(\ref{eq:sampling}) by the above construction, we show in Appendix \ref{app:approx} that for any \mbox{$\epsilon>0$}, one can pick approximation parameters $\xi_{k;j}=\text{poly }(\epsilon,\frac1m)$ for all $k\in\{1,\dots,m\}$ and all $j\in\{1,\dots,r^\star(y_k)\}$, such that:
\begin{equation}\label{eq:samplingAux}\begin{split}
   & P(y_1,\dots, y_m \lvert \hat \rho) \\&\!\!=\!\frac{1}{{\cal N}\prod_{k=1}^m\!\prod_{j=1}^{r^\star(y_k)}\xi_{k;j}^2} \tr \left[ \hat \rho_{\rm total} \left(\bigotimes_{k=1}^m\ket{\alpha_k}\!\bra{\alpha_k}\otimes \ket0\!\bra0^{\otimes n}\!\right) \right]\!+\!\mathcal O(\epsilon),
    \end{split}
\end{equation}
where we have set $n:=\sum_{k=1}^mr^\star(y_k)$, and where the state $\hat \rho_{\rm total}$ is defined on the full Hilbert space, including all the auxiliary modes, and is given by
\begin{equation}
   \hat \rho_{\rm total}:=(\hat{\bm S}^{\dag}\otimes\hat{\mathds{1}}_{\rm aux})\hat {\cal U}^{\dag}\left(\hat\rho \otimes \ket1\!\bra1^{\otimes n}\right)\hat{\cal U}(\hat{\bm S}\otimes\hat{\mathds{1}}_{\rm aux}),
\end{equation}
with $\hat {\cal U}$ given by $\hat {\cal U}:=\bigotimes_{k=1}^m\prod_{j=1}^{r^{\star}(y_k)} \hat D(\beta_{k;j}) \hat U^{\dag}(\xi_{k;j}) \hat D^{\dag}(\beta_{k;j})$. We note that $\smash{\hat U(\xi_{k;j})}$ is the two-mode squeezer that connects the $k^{th}$ detection mode to the auxiliary mode that implements the $j^{th}$ photon-subtraction operation associated with it and thus $(\hat{\bm S}^{\dag}\otimes \hat{\mathds{1}}_{\rm aux})\hat {\cal U}^{\dag}$ is a Gaussian unitary. In particular, $r^\star(\hat\rho_{\rm total})=r^\star(\hat\rho\otimes\ket1\!\bra1^{\otimes n})=r^\star(\hat\rho)+\sum_{k=1}^mr^\star(y_k)$ since the stellar rank is fully additive with respect to tensor products with pure states~\cite{chabaud2021holomorphic}.

\begin{figure}
	\begin{center}
		\includegraphics[width=\columnwidth]{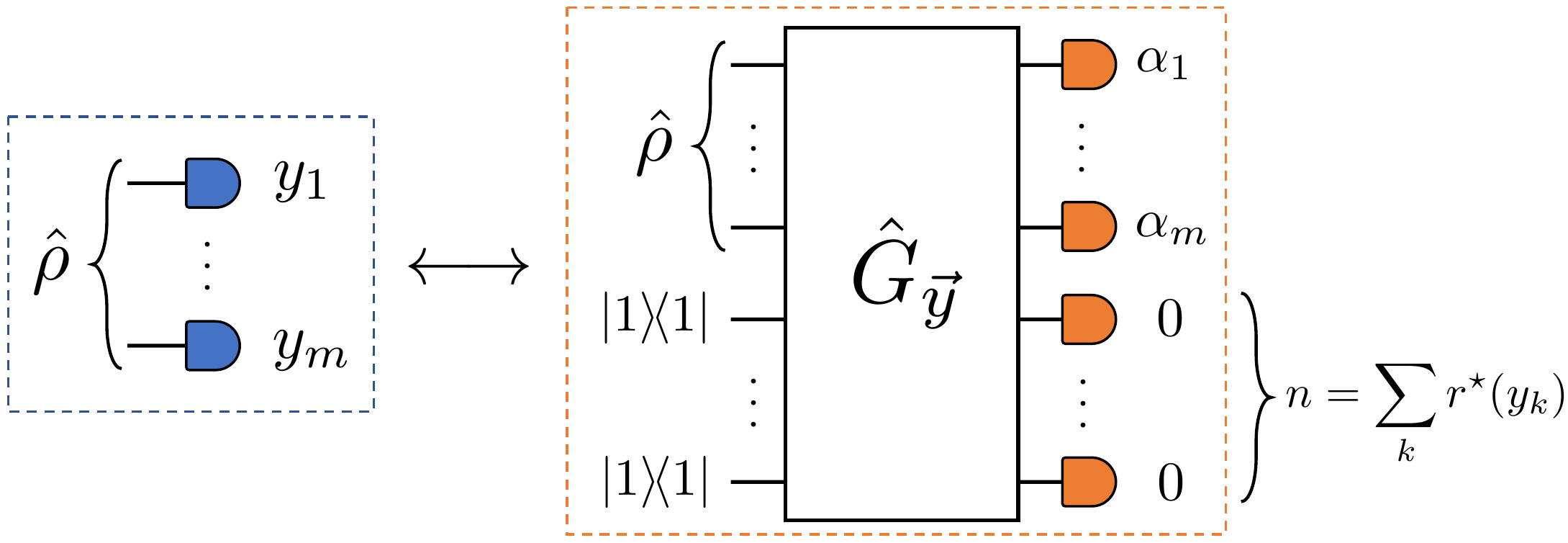}
    		\caption{To any bosonic computation (left, in blue) is associated a coherent state sampling setup (right, in orange) which takes as input the same state $\hat\rho$, together with auxiliary single-photon Fock states, and whose output probability density approximates to arbitrary precision the output probability of a given outcome up to normalisation, i.e.\ $P(\alpha_1,\dots,\alpha_m,0,\dots,0)\approx\frac1{\mathcal N'}P(y_1,\dots,y_m)$. The number of auxiliary Fock states $n$ is the sum of the stellar ranks of the projectors associated with the outcomes $y_1,\dots,y_m$.}
		\label{fig:dual}
	\end{center}
\end{figure}

The projection on the vacuum is consistent with double homodyne detection since $\ket{0}_{vac}$ is also a coherent state. The expression in Eq.~\eqref{eq:samplingAux} thus shows that any setup where one samples a given outcome from a bosonic state can be mapped theoretically to a larger coherent state sampling setup, whose output probability density matches to arbitrary precision the output probability of that outcome, up to a normalising factor (see Fig.~\ref{fig:dual}). Furthermore, the stellar ranks of the projection operators translate to the inclusion of additional single-photon Fock states in auxiliary modes. A similar derivation, detailed in Appendix \ref{app:marginals}, shows that the corresponding marginal probabilities are also reproduced by the marginals probability densities of coherent state samplers.

\textit{Strong simulation of bosonic computations.---}These results highlight that coherent state samplers can be very generally used to simulate other sampling setups using similar techniques as in~\cite{chabaud2020classical}. Strong simulation in particular refers to the evaluation of any output probability of a computation, or any of its marginals probabilities. Hereafter, we rely on the following notion of approximate strong simulation: let $P$ be a probability distribution (density); for $\epsilon>0$, approximate strong simulation of $P$ up to total variation distance $\epsilon$ refers to the computational task of strongly simulating a probability distribution $Q$ which is $\epsilon$-close to $P$ in total variation distance (see Appendix \ref{app:strong} for a formal definition).

The classical algorithm for strong simulation of Gaussian circuits with non-Gaussian input states from~\cite[Theorem 2]{chabaud2020classical} can be readily applied to coherent state samplers. Combining this result with our construction, we obtain a general classical algorithm for approximate strong simulation of bosonic quantum computations whose complexity scales with the stellar rank of both the input state and the measurement setup. We state the result in the case of pure state input and projective measurements and refer to Theorem 2 in Appendix \ref{app:strong} for the general theorem and its proof:

\textit{Theorem 1.---}Let $\ket{\bm\psi}$ be an $m$-mode pure state of stellar rank $r^\star(\bm\psi)$ and core state support size $s$. For all $k\in\{1,\dots,m\}$, let $\hat Y_k$ be an observable with eigenbasis $\{\ket{y_k}\}_{y_k\in\mathcal Y_k}$, and let $r^\star_k=\sup_{y_k\in\mathcal Y_k}r^\star(y_k)$. Let $r:=\smash{r^\star_{\bm\psi}+\sum_k}r^\star_k$ be the total stellar rank of the setup. Then, the measurement of $\hat Y_1,\dots,\hat Y_m$ on $\ket{\bm\psi}$ over an exponentially large outcome space can be approximately strongly simulated up to total variation distance $\mathrm{exp}(-\text{poly }m)$ in time ${\cal O}(s^2r^32^r+\text{poly }m)$.

The total variation distance in the theorem results from the approximation used in Eq.~(\ref{eq:samplingAux}). This strong simulation algorithm competes with state-of-the art classical algorithms for certain bosonic architectures~\cite{chabaud2020classical}, but applies to a much wider class of quantum computations---essentially any bosonic computation. The time complexity in Theorem~1 is a worst-case complexity, based on the fastest known classical algorithm for computing the hafnian~\cite{bjorklund2019faster}, and may be reduced for particular instances. On the other hand, due to its broad applicability, our simulation technique may be outperformed by classical simulation algorithms targeting specific classes of bosonic circuits~\cite{quesada2018gaussian,garcia2019simulating,renema2020simulability,bourassa2021fast,calcluth2022efficient}. Nonetheless, Theorem~1 may be used primarily as a tool for identifying necessary resources for bosonic quantum computational advantage: it establishes the stellar rank as a necessary non-Gaussian property.

\textit{Non-Gaussian entanglement.---}Now that we have shown that any bosonic computation can be connected to a coherent state sampler, we aim to identify physical resources that are required to reach a quantum advantage with coherent state sampling beyond the stellar rank. We resort to a basic model of coherent state sampler, where we consider sampling from a given $N$-mode state $\hat \sigma$. The probability density corresponding to a certain set of complex measurement outcomes $\alpha_1,\dots,\alpha_N$ in the $N$ output detectors is given by the Husimi $Q$-function of the state $\hat \sigma$: $Q(\vec \alpha \lvert \hat \sigma) = \frac{1}{\pi^N}\bra{\vec \alpha} \hat \sigma \ket{\vec \alpha}$, where $\vec\alpha=(\alpha_1,\dots,\alpha_N)^\top$. By having put all the quantum resources of the sampling protocol at the level of the state, the hardness of the sampling problem can now be directly related to properties of the resourceful state's $Q$-function.

Under basic assumptions, we can efficiently sample classically from the $Q$-function of any separable mixed state (see Appendix \ref{app:Qsamp} for a discussion). Hence, quantum entanglement of the input state is a necessary requirement in the design of a coherent state sampler that is hard to simulate. However, it turns out that not all forms of entanglement are equally suitable. In previous works \cite{walschaers_entanglement_2017,chabaud2021holomorphic}, we have discussed the concept of passive separability: a quantum state is said to be passively separable if at least one mode-basis exists in which the state is separable. In other words, for a passively separable state, any entanglement can be undone by an interferometer built with beam-splitters and phase-shifters. 

The concept of passive separability becomes essential when we combine it with the properties of coherent states. Let $\hat U$ describe a passive $N$-mode linear optics interferometer in the sense that $\hat U^\dag \hat a_k \hat U = \sum_{j} U_{jk}\hat a_j$, where $U$ is an $N\times N$ unitary matrix. The action of $\hat U$ on an $N$-mode coherent state is given by $\hat U \ket{\vec \alpha} = \ket{U\vec \alpha}$. This simple identity implies that for all passive linear optics transformations, $Q(\vec \alpha \lvert  \hat \tau) = Q(U \vec \alpha \lvert \hat U \hat \tau\hat U^{\dag})$.
By definition, for any state $\hat \tau$ which is passively separable, there is at least one transformation $\hat U$ such that $\hat U \hat \tau \hat U^{\dag}$ is separable. This, in turn, means that we can efficiently sample from the distribution $Q(\vec \alpha \lvert \hat U \hat \tau \hat U^{\dag})$. Hence, we can sample a vector $\vec \alpha$ from $Q(\vec \alpha \lvert  \hat \tau)$ by first sampling $\vec \beta$ distributed according to $Q(\vec \beta \lvert \hat U \hat \tau \hat U^{\dag})$ and subsequently identifying $\vec \alpha = U^{\dag} \vec \beta$. Thus, we find that we can efficiently simulate the coherent state sampling from any passively separable state.

To reach a quantum computational advantage with a coherent state sampler, we thus have to use input states that are not passively separable. This requirement immediately excludes all Gaussian states, since these are always passively separable \cite{PhysRevA.96.053835}. The lack of passive separability can therefore be seen as non-Gaussian entanglement in the sense that it is a form of entanglement that persists in any mode-basis and cannot be extracted based solely on the state's covariance matrix. It thus highlights the presence of non-Gaussian features in the state's correlations.

We emphasize that there are other intuitive notions of non-Gaussian entanglement. When we call states that are separable through general Gaussian operations (i.e.\ a combination of interferometers and squeezing operations) Gaussian-separable, one could say that only states which are not Gaussian-separable have non-Gaussian entanglement. To understand what notion of non-Gaussian entanglement is necessary for reaching a quantum computational advantage with coherent state sampling, we consider the seminal example of Boson Sampling. Through Eq.~\eqref{eq:samplingAux}, we find that ideal Boson Sampling with $n$ input photons and an $m$-mode interferometer $\hat U_{\rm BS}$ corresponds to coherent state sampling from a state given by $\ket{\Psi} \propto  \hat{\cal U}\left(\hat U_{\rm BS} \otimes \mathds{1}_{\rm aux}\right)\ket{\Psi_{\rm total}}$, where $\hat{\cal U}$ is a tensor product of two-mode squeezers and where the state $\ket{\Psi_{\rm total}}$ is a $2n$-photon Fock state that combines the input state of the boson sampler with $n$ auxiliary photons, given by
\begin{equation}
\!\ket{\Psi_{\rm total}}\!=\!\big[\underbrace{\ket{1}\otimes\dots \otimes\ket{1}}_{n}\otimes \underbrace{\ket{0}\otimes  \dots \otimes \ket{0}}_{m-n}\big] \otimes \big[\underbrace{\ket{1}\otimes\dots \otimes\ket{1}}_{n}\big]_{\rm aux}.  
\end{equation}
Boson Sampling is known to be a hard problem, so exact coherent state sampling from the state $\ket{\Psi}$ is also classically hard~\footnote{When multiplicative estimation or additive estimation up to exponentially small error of a single outcome probability is an instance of a $\#$\textsf{P}-hard problem, which is the case for Boson Sampling, the use of Stockmeyer's algorithm implies that any efficient classical simulation of exact sampling would collapse the polynomial hierarchy of complexity classes~\cite{10.1145/1993636.1993682,hangleiter2022computational}. This argument directly applies to the corresponding instance of coherent state sampling}. The structure of this state nicely highlights the three fundamental types of non-classicality that are required: non-Gaussian resources in $\ket{\Psi_{\rm total}}$, large-scale entanglement through $\hat U_{\rm BS} $, and squeezing through $\hat{\cal U}$. Furthermore, the order of the elements is essential: the state $\ket{\Psi}$ is not passively separable because the squeezing operations in $\hat{\cal U}$ and the non-Gaussian features in $\ket{\Psi_{\rm total}}$ are local in a different mode basis. However, $\hat{\cal U}\left(\hat U_{\rm BS} \otimes \mathds{1}_{\rm aux}\right)$ is a Gaussian operation and $\ket{\Psi_{\rm total}}$ is separable. This means that the state $\ket{\Psi}$ is thus Gaussian-separable but not passive-separable. Hence, there are Gaussian-separable states leading to coherent state sampling that cannot be efficiently simulated. We thus propose to define non-Gaussian entanglement as the type of entanglement that is present in states that are not passively separable. This amounts to defining it operationally as a type of entanglement that is necessary to achieve computationally hard coherent state sampling.

\textit{Conclusion.---}In this work, we argue that any bosonic sampling computation can be mapped to a corresponding coherent state sampling computation. Our construction allows us to derive a general classical algorithm for strong simulation of bosonic computations, whose time complexity scales with the stellar rank of the input state and the measurement setup of the computation.

We see our work in first instance as providing a useful method to analyse the resources in sampling setups because all resources in coherent state sampling are situated at the level of the state. As such, we also find that coherent state sampling with passively separable states can be simulated efficiently. We therefore find that the lack of passive separability rather than the lack of Gaussian separability is the operationally useful type of non-Gaussian entanglement.

Our key reduction in Eq.~\eqref{eq:samplingAux} shows that any non-Gaussian resource in the measurement is introduced in the coherent state sampler through auxiliary photons. The total number of auxiliary photons in the coherent state sampler ultimately corresponds to the total stellar rank of the measurement setup. These photons must be entangled in a fundamentally non-Gaussian way to achieve the necessary sampling complexity. For pure states, this non-Gaussian entanglement also implies one of the previous requirements for reaching a quantum computational advantage: Wigner negativity \cite{PhysRevLett.109.230503}. Yet, for mixed states it remains an open question how the necessity of Wigner negativity translates to the coherent state sampler.

Typical sampling setups such as (Gaussian) Boson Sampling correspond to reasonably simple coherent state samplers that mix local non-Gaussian resources through a multimode Gaussian transformation. However, in the multimode bosonic state space much more exotic states can be conceived. Preparing such states would require multimode non-Gaussian unitary transformations, and it would be interesting to understand whether they have any additional computational resourcefulness.\\

\textit{Acknowledgements.---}We thank Fr\'ed\'eric Grosshans for inspiring discussions. This work was supported by the ANR JCJC project NoRdiC (ANR-21-CE47-0005) and Plan France 2030 project NISQ2LSQ (ANR-22-PETQ-0006). UC acknowledges funding provided by the Institute for Quantum Information and Matter, an NSF Physics Frontiers Center (NSF Grant PHY-1733907).

\bibliography{QuantumAdvantage,biblio}

\onecolumngrid

\newpage

\appendix

\section{Preliminaries}
\label{app:stellar}

\subsection{Multi-index notations}

We use bold math to denote multimode/multivariate notations.
For $m\in\mathbb N^*$, $\bm p=(p_1,\dots,p_m)\in\mathbb N^m$, $\bm q=(q_1,\dots,q_m)\in\mathbb N^m$, $\bm z=(z_1,\dots,z_m)\in\mathbb C^m$, and $\bm w=(w_1,\dots,w_m)\in\mathbb C^m$ we write:
\begin{equation}
    \begin{aligned}
    \bm0&=(0,\dots,0)\\
    |\bm p|&=p_1+\dots+p_m\\
    \bm p!&=p_1!\dots p_m!\\
    \ket{\bm p}&=\ket{p_1}\otimes\dots\otimes\ket{p_m}\\
    \bm p\le\bm q&\Leftrightarrow\forall k\in\{1,\dots,m\},\;p_k\le q_k\\
    \bm z^*&=(z_1^*,\dots,z_m^*)\\
    \|\bm z\|^2&=|z_1|^2+\dots+|z_m|^2\\
    \bm z^{\bm p}&=z_1^{p_1}\dots z_m^{p_m}\\
    d^{2m}\bm z&=d\mathcal R(z_1)d\mathcal I(z_1)\dots d\mathcal R{(z_m)}d\mathcal I(z_m)\\
    \end{aligned}
\end{equation}

\subsection{Stellar formalism}

The stellar formalism has been introduced in~\cite{chabaud2020stellar} to characterize single-mode non-Gaussian quantum states. It has later been generalised to the multimode setting in~\cite{chabaud2020classical,chabaud2021holomorphic}. In this section, we give a brief review of this formalism.

\medskip

The stellar formalism classifies non-Gaussian quantum states using the so-called \textit{stellar function}, which is a representation of a quantum state in an infinite-dimensional Hilbert space by a holomorphic function originally due to Segal~\cite{segal1963mathematical} and Bargmann~\cite{bargmann1961hilbert}.
More precisely, to any $m$-mode quantum state $\ket{\bm\psi}=\sum_{\bm n\ge\bm 0}\psi_{\bm n}\ket{\bm n}$ is associated its stellar function
\begin{equation}
    F^\star_{\bm\psi}(\bm z):=\sum_{\bm n\ge\bm 0}\frac{\psi_{\bm n}}{\sqrt{\bm n!}}\bm z^{\bm n},
\end{equation}
for all $\bm z\in\mathbb C^m$. The stellar function is related to the Husimi $Q$-function of the state as
\begin{equation}
    Q_{\bm\psi}(\bm\alpha)=\frac{e^{-\frac12\|\bm\alpha\|^2}}{\pi^m}\left|F^\star_{\bm\psi}(\bm\alpha^*)\right|^2,
\end{equation}
for all $\bm\alpha\in\mathbb C^m$. 

As it turns out, the stellar function has zeros if and only if the corresponding state is non-Gaussian.
These zeros allow us to rank non-Gaussian states: in the single-mode setting, the \textit{stellar rank} of a pure quantum state is defined as the number of zeros of its stellar function. Equivalently, single-mode pure quantum states of finite stellar rank $r^\star\in\mathbb N$ are those states whose stellar function is of the form $P\times G$, where $P$ is a polynomial of degree $r^\star$ and $G$ is a Gaussian function, while all other states are of infinite stellar rank. 
This equivalent definition allows us to define the stellar rank in the multimode setting, where a characterization in terms of number of zeros is no longer possible: multimode pure quantum states of finite stellar rank $r^\star\in\mathbb N$ are those states whose stellar function is of the form $P\times G$, where $P$ is a \textit{multivariate} polynomial of degree $r^\star$ and $G$ is a \textit{multivariate} Gaussian function, while all other states are of infinite stellar rank.
For a mixed state $\hat\rho$, the stellar rank is defined using a convex roof construction: $\smash{r^\star(\hat\rho)=\inf_{p_i,\bm\psi_i}\sup r^\star_{\bm\psi_i}}$, where the infinimum is over the decompositions $\hat\rho=\smash{\sum_ip_i\ket{\bm\psi_i}\!\bra{\bm\psi_i}}$.

The stellar rank induces a non-Gaussian hierarchy, the \textit{stellar hierarchy}, among quantum states. At rank $0$ lie mixtures of Gaussian states, while quantum non-Gaussian states populate all higher ranks. For example, a multimode Fock state $\ket{\bm n}$ has stellar rank $|\bm n|$, while a single-mode cat state has infinite stellar rank.

Additionally, in the main text we refer to the stellar rank of a measurement setup: we define naturally the stellar rank of an observable $\hat Y$ as being the supremum of the stellar ranks of its eigenvectors.

\medskip

The stellar function, stellar rank and stellar hierarchy possess various remarkable properties~\cite{chabaud2020stellar,chabaud2020certification,chabaud2021holomorphic}:
\begin{enumerate}[label=(\roman*)]
    \item\label{enum:i} The stellar function gives a unique prescription for engineering a pure quantum state from the vacuum by applying a generating function of the creation operators of the modes: $\ket{\bm\psi}=F^\star_{\bm\psi}(\hat a^\dag_1,\dots,\hat a^\dag_m)\ket{\bm0}$.
    \item\label{enum:ii} The form $P\times G$ of the stellar function translates to the structure of the corresponding quantum state: finite stellar rank states have a normal expansion $P(\hat a^\dag_1,\dots,\hat a^\dag_m)\ket G$ where $P$ is a polynomial and $\ket G$ is a Gaussian state. 
    \item\label{enum:iii} A unitary  operation  is  Gaussian  if  and  only  if  it  leaves  the stellar  rank invariant, and the stellar rank is non-increasing under Gaussian channels and measurements. As a result, the stellar rank is a non-Gaussian monotone which provides a measure of the non-Gaussian character of a quantum state. Moreover, the stellar rank is additive with respect to a tensor product with a pure state.
    \item\label{enum:iv} A state of stellar rank $n$ cannot be obtained from the vacuum by using less than $n$ applications of creation operators, together with Gaussian unitary operations. In other words, the stellar rank has an operational interpretation relating to point~\ref{enum:i}, as a lower bound on the minimal number of elementary non-Gaussian operations (applications of a single-mode creation operator, i.e.\ single-photon addition) needed to engineer the state from the vacuum---a notion which may be thought of as a bosonic counterpart to the $T$-count of finite-dimensional quantum circuits~\cite{amy2013meet,beverland2020lower}.
    \item\label{enum:v} The structure from point~\ref{enum:ii} can be reverted, yielding another useful characterisation: finite stellar rank states also have an antinormal expansion $\hat G\ket C$, where $\hat G$ is a Gaussian unitary and $\ket C$ is a core state, i.e.\ a state with finite support over the Fock basis. By extension, the \textit{support size} of a state of finite stellar rank refers to the size of the support of its core state over the Fock basis. 
    \item\label{enum:vi} The stellar hierarchy is robust with respect to the trace distance, i.e.\ every state of a given finite stellar rank only has states of equal or higher rank in its close vicinity. As a consequence, the stellar rank of quantum states can be witnessed experimentally.
    \item\label{enum:vii} States of finite stellar rank form a dense subset of the Hilbert space, so that states of infinite rank can be approximated arbitrarily well in trace distance by states of finite stellar rank. Moreover, the optimal approximation of any fixed rank can be obtained by an optimization over ${\cal O}(m^2)$ parameters, independently of the rank.
\end{enumerate}

\section{Photon subtraction and addition using an auxiliary photon}\label{app:photsub}

In this section we show how to retrieve photon subtraction by using an auxiliary single-photon state and weak two-mode squeezing, together with a projection of the auxiliary mode on the vacuum.

By linearity, we only treat the case of an input pure state $\ket\psi$, i.e.\ we compute the output state of the Completely Positive Trace-Decreasing (CPTD) map
\begin{equation}\label{SM:Exi}
    \mathcal E_\xi(\ket\psi):=\tr_{\rm aux} \left\{\hat U_{\rm TMSS}(\xi) [\ket\psi\!\bra\psi\otimes \ket{1}\!\bra{1}_{\rm aux}]\hat U_{\rm TMSS}(\xi)^\dag [\hat{\mathds{1}}\otimes\ket{0}\!\bra{0}_{\rm aux}]\right\},
\end{equation}
where $\hat U_{\rm TMSS}(\xi)=\exp[\xi\hat a^{\dag}\hat a_{\rm aux}^{\dag}-\xi^*\hat a\hat a_{\rm aux}]$ is the two-mode squeezing operator. Writing $\xi=re^{i\theta}$ and $\mu=\cosh r$, $\nu=e^{i\theta}\sinh r$, the normal ordering of the two-mode squeezing operator is given by~\cite{hong1987new}:
\begin{equation}
    \hat U_{\rm TMSS}(\xi)=e^{\frac\nu\mu\hat a^\dag\hat a_{\rm aux}^\dag}\mu^{-(\hat a^\dag\hat a+\hat a_{\rm aux}^\dag\hat a_{\rm aux}+\hat{\mathds{1}})}e^{-\frac{\nu^*}\mu\hat a\hat a_{\rm aux}}.
\end{equation}
We have
\begin{equation}
    e^{-\frac{\nu^*}\mu\hat a\hat a_{\rm aux}}(\ket\psi\otimes\ket1_{\rm aux})=\ket\psi\otimes\ket1_{\rm aux}-\frac{\nu^*}\mu\hat a\ket\psi\otimes\ket0_{\rm aux},
\end{equation}
by expanding the exponential series (only the first two terms yield a nonzero contribution). Similarly, for all $\ket\chi$ we have
\begin{equation}
    \bra\chi\otimes\bra0_{\rm aux}e^{\frac\nu\mu\hat a^\dag\hat a_{\rm aux}^\dag}\mu^{-(\hat a^\dag\hat a+\hat a_{\rm aux}^\dag\hat a_{\rm aux}+\hat{\mathds{1}})}=\bra\chi\mu^{-(\hat a^\dag\hat a+1)}\otimes\bra0_{\rm aux}.
\end{equation}
Hence,
\begin{equation}\label{SM:devExi}
    \begin{aligned}
        \mathcal E_\xi(\ket\psi)&=\bra0_{\rm aux}\hat U_{\rm TMSS}(\xi)\left(\ket\psi\otimes\ket1_{\rm aux}\right)\\
        &=\bra0_{\rm aux}e^{\frac\nu\mu\hat a^\dag\hat a_{\rm aux}^\dag}\mu^{-(\hat a^\dag\hat a+\hat a_{\rm aux}^\dag\hat a_{\rm aux}+\hat{\mathds{1}})}\cdot e^{-\frac{\nu^*}\mu\hat a\hat a_{\rm aux}}(\ket\psi\otimes\ket1_{\rm aux})\\
        &=\mu^{-(\hat a^\dag\hat a+\hat{\mathds{1}})}\bra0_{\rm aux}\left(\ket\psi\otimes\ket1_{\rm aux}-\frac{\nu^*}\mu\hat a\ket\psi\otimes\ket0_{\rm aux}\right)\\
        &=-\nu^*\mu^{-2}\mu^{-\hat a^\dag\hat a}\hat a\ket\psi\\
        &=\left(-\frac{e^{-i\theta}\sinh r}{\cosh^2 r}\right)(\cosh\xi)^{-\hat a^\dag\hat a}\hat a\ket\psi.
    \end{aligned}
\end{equation}
In the limit $r\rightarrow0$ we retrieve $\ket\phi\!\bra\phi\sim r^2\hat a\ket\psi\!\bra\psi\hat a^\dag$, so the map $\ket\psi\mapsto\left(-\frac{\cosh^2 r}{e^{-i\theta}\sinh r}\right)\mathcal E_\xi(\ket\psi)=(\cosh\xi)^{-\hat a^\dag\hat a}\hat a\ket\psi$ with the removed prefactor approximates a single-photon subtraction. We make this statement more precise in the following section. Note that for $r\neq0$, the state $\ket\phi$ is an attenuated photon-subtracted state. 

\medskip

The same arguments may be used to show that photon addition may be performed by using an auxiliary single-photon state and a beam splitter operation $\hat U_{\rm BS}(\zeta)=e^{\zeta\hat a^\dag\hat a_{\rm aux}-\zeta^*\hat a\hat a_{\rm aux}^\dag}=e^{-e^{i\delta}\tan\gamma\hat a\hat a^\dag_{\rm aux}}(\cos^2\gamma)^{\hat a^\dag\hat a-\hat a_{\rm aux}^\dag\hat a_{\rm aux}}e^{e^{i\delta}\tan\gamma\hat a^\dag\hat a_{\rm aux}}$~\cite{ferraro2005gaussian}, with $\zeta=\gamma e^{i\delta}$, together with a projection of the auxiliary mode on the vacuum. The state obtained is an attenuated photon-added state:
\begin{equation}
    \left(e^{i\delta}\tan\gamma\right)(\cos\gamma)^{2\hat a^\dag\hat a}\hat a^\dag\ket\psi.
\end{equation}

\medskip

For completeness, we consider the case of finite-resolution postselection, where the auxiliary output is measured with double homodyne detection, yielding a complex outcome $\alpha\in\mathbb C$ (with possibly $\alpha\neq0$). 
For all $\ket\chi$ we have
\begin{equation}
    \begin{aligned}
        \bra\chi\otimes\bra\alpha_{\rm aux}e^{\frac\nu\mu\hat a^\dag\hat a_{\rm aux}^\dag}\mu^{-(\hat a^\dag\hat a+\hat a_{\rm aux}^\dag\hat a_{\rm aux}+\hat{\mathds{1}})}&=\bra\chi e^{\frac\nu\mu\alpha^*\hat a^\dag}\mu^{-(\hat a^\dag\hat a+\hat{\mathds{1}})}\otimes\bra\alpha_{\rm aux}\mu^{-\hat a_{\rm aux}^\dag\hat a_{\rm aux}}\\
        &=\bra\chi e^{\frac\nu\mu\alpha^*\hat a^\dag}\mu^{-(\hat a^\dag\hat a+\hat{\mathds{1}})}\otimes e^{-\frac12|\alpha|^2}e^{\frac12|\alpha/\mu|^2}\bra{\alpha/\mu}_{\rm aux}.
    \end{aligned}
\end{equation}
For the two-mode squeezing gadget we obtain:
\begin{equation}
    \begin{aligned}
        \bra\alpha_{\rm aux}\hat U_{\rm TMSS}(\xi)\left(\ket\psi\otimes\ket1_{\rm aux}\right)&=\bra\alpha_{\rm aux}e^{\frac\nu\mu\hat a^\dag\hat a_{\rm aux}^\dag}\mu^{-(\hat a^\dag\hat a+\hat a_{\rm aux}^\dag\hat a_{\rm aux}+\hat{\mathds{1}})}\cdot e^{-\frac{\nu^*}\mu\hat a\hat a_{\rm aux}}(\ket\psi\otimes\ket1_{\rm aux})\\
        &=e^{\frac\nu\mu\alpha^*\hat a^\dag}\mu^{-(\hat a^\dag\hat a+\hat{\mathds{1}})} e^{-\frac12|\alpha|^2}e^{\frac12|\alpha/\mu|^2}\bra{\alpha/\mu}_{\rm aux}\left(\ket\psi\otimes\ket1_{\rm aux}-\frac{\nu^*}\mu\hat a\ket\psi\otimes\ket0_{\rm aux}\right)\\
        &=\frac1{\mu^2}e^{-\frac12|\alpha|^2}e^{\frac\nu\mu\alpha^*\hat a^\dag}\mu^{-\hat a^\dag\hat a} \left(\alpha^*\hat{\mathds1}-\nu^*\hat a\right)\ket\psi.
    \end{aligned}
\end{equation}
We retrieve a photon subtraction for $|\xi|\rightarrow0$ and $|\alpha|=o(|\xi|)$.
Similarly, for the beam splitter gadget we obtain:
\begin{equation}
    \frac1{\cos^2\gamma}e^{-\frac12|\alpha|^2}e^{-e^{i\delta}\tan\gamma\alpha^*\hat a}(\cos\gamma)^{2\hat a^\dag\hat a}\left(\alpha^*\hat{\mathds 1}+e^{i\delta}\cos\gamma\sin\gamma\hat a^\dag\right)\ket\psi.
\end{equation}
We retrieve a photon addition for $|\zeta|\rightarrow0$ and $|\alpha|=o(|\zeta|)$.

\section{Approximate probability distribution}
 \label{app:approx}

In this section, we analyze the error introduced at the level of the probability distribution by using the auxiliary photon gadgets from the previous section instead of the photon subtractions, thus justifying Eq.~(3) in the main text. 

\medskip

Hereafter, we denote the vector $2$-norm of a pure state $\ket\psi$ by $\|\ket\psi\|=\sqrt{\langle\psi|\psi\rangle}$, and the trace distance between two pure states $\ket\phi$ and $\ket\psi$ is given by
\begin{equation}
    D(\ket\phi,\ket\psi)=\frac12\tr(|\ket\phi\!\bra\phi-\ket\psi\!\bra\psi|)=\frac12\|\ket\phi\!\bra\phi-\ket\psi\!\bra\psi|\|_1,
\end{equation}
where $\|\cdot\|_1$ is the operator $1$-norm (note that $\|\ket\psi\!\bra\psi\|_1=\|\ket\psi\|^2$).

Recalling the notations of the main text, the output projectors $\hat P_{k;y_k}$ correspond to finite stellar rank pure states $\ket{y_k}$ of the form
\begin{equation}\label{SM:defyk}
    \ket{y_k}=\frac1{\sqrt{{\cal N}_k}} \left[\prod_{j=1}^{r^\star(y_k)} \hat D(\beta_{k;j}) \hat a^{\dag}_k \hat D^{\dag}(\beta_{k;j})\right]\hat S_k\ket{\alpha_k},
\end{equation}
where $\smash{r^\star(y_k)}\in\mathbb N$ denotes the stellar rank of the state $\ket{y_k}$, $\hat D(\beta_{k;j})$ is a displacement operator that acts on mode $k$ with $\beta_{k;j}\in\mathbb C$, $\smash{\hat a^{\dag}_k}$ is the creation operator in mode $k$, $\hat S_k$ is a squeezing operator, $\ket{\alpha_k}$ is a coherent state, and ${\cal N}_k$ is a normalisation factor, for all $k\in\{1,\dots,m\}$ and all $j\in\{1,\dots,r^\star(y_k)\}$.
The output probability distribution is then given by
\begin{equation}\label{SM:outputprobaexpand}
    \begin{aligned}
        P(y_1,\dots, y_m \lvert\hat \rho)&=\tr\left[\hat\rho\bigotimes_{k=1}^m\ket{y_k}\!\bra{y_k}\right]\\
        &=\frac1{\cal N}\tr\left[\hat\rho\left(\bigotimes_{k=1}^m\prod_{j=1}^{r^\star(y_k)}\hat D(\beta_{k;j})\hat a_k^\dag\hat D^{\dag}(\beta_{k;j})\hat S_k\right)\bigotimes_{k=1}^m\ket{\alpha_k}\!\bra{\alpha_k}\left(\bigotimes_{k=1}^m\hat S_k^\dag\prod_{j=1}^{r^\star(y_k)}\hat D(\beta_{k;j})\hat a_k\hat D^{\dag}(\beta_{k;j})\right)\right]\\
        &=\frac1{\cal N}\tr\left[\left(\bigotimes_{k=1}^m\hat S_k^\dag\prod_{j=1}^{r^\star(y_k)}\hat D(\beta_{k;j})\hat a_k\hat D^{\dag}(\beta_{k;j})\right)\hat\rho\left(\bigotimes_{k=1}^m\prod_{j=1}^{r^\star(y_k)}\hat D(\beta_{k;j})\hat a_k^\dag\hat D^{\dag}(\beta_{k;j})\hat S_k\right)\bigotimes_{k=1}^m\ket{\alpha_k}\!\bra{\alpha_k}\right],
    \end{aligned}
\end{equation}
where ${\cal N}=\prod_{k=1}^m{\cal N}_k$ is a normalisation factor. 

The probability distribution obtained by replacing each photon subtraction $\hat a$ in the above expression by $(\cosh\xi)^{-\hat a^\dag\hat a}\hat a$ for some $\xi\in\mathbb R$ (which by Eq.~(\ref{SM:devExi}) is the map $\mathcal E_\xi$ with the prefactor $-\frac{\sinh \xi}{\cosh^2 \xi}$ removed) reads
\begin{equation}\label{SM:tildePxi}
    \begin{aligned}
        \tilde P_{\bm\xi}(y_1,\dots, y_m \lvert\hat \rho)&=\frac1{\cal N}\tr\left[\left(\bigotimes_{k=1}^m\hat S_k^\dag\prod_{j=1}^{r^\star(y_k)}\hat D(\beta_{k;j})(\cosh\xi_{k;j})^{-\hat a^\dag\hat a}\hat a_k\hat D^{\dag}(\beta_{k;j})\right)\hat\rho\left(\bigotimes_{k=1}^m\prod_{j=1}^{r^\star(y_k)}\hat D(\beta_{k;j})\hat a_k^\dag(\cosh\xi_{k;j})^{-\hat a^\dag\hat a}\hat D^{\dag}(\beta_{k;j})\hat S_k\right)\bigotimes_{k=1}^m\ket{\alpha_k}\!\bra{\alpha_k}\right]\\
        &=\frac1{\cal N}\prod_{k=1}^m\prod_{j=1}^{r^\star(y_k)}\left(\frac{\cosh^2(\xi_{k;j})}{\sinh (\xi_{k;j})}\right)^2\!\!\tr\Bigg[\left(\bigotimes_{k=1}^m\hat S_k^\dag\prod_{j=1}^{r^\star(y_k)}\hat D(\beta_{k;j})\hat U_{\rm TMSS}(\xi_{k;j})\hat D^{\dag}(\beta_{k;j})\right)\left(\hat\rho\otimes\ket1\!\bra1^{\otimes n}\right)\\
        &\quad\quad\quad\quad\quad\quad\quad\quad\quad\quad\quad\quad\quad\quad\times\left(\bigotimes_{k=1}^m\prod_{j=1}^{r^\star(y_k)}\hat D(\beta_{k;j})\hat U_{\rm TMSS}(\xi_{k;j})^\dag\hat D^{\dag}(\beta_{k;j})\hat S_k\right)\left(\bigotimes_{k=1}^m\ket{\alpha_k}\!\bra{\alpha_k}\otimes\ket0\!\bra0^{\otimes n}\right)\Bigg],
    \end{aligned}
\end{equation}
where we used Eq.~(\ref{SM:Exi}) in the last line, where for all $k\in\{1,\dots,m\}$ and all $j\in\{1,\dots,r^\star(y_k)\}$, $\xi_{k;j}\in\mathbb R$ and $\hat U_{\rm TMSS}(\xi_{k;j})$ is the two-mode squeezer that implements the $j^{th}$ photon subtraction for the $k^{th}$ mode~\footnote{Note that the elements of the product no longer commute. We assume that the ordering is fixed.}, and where we have set $n:=\sum_{k=1}^mr^\star(y_k)$.

The main result of this section is the following:

\begin{theorem}\label{th:approx}
Let $\epsilon>0$. For all $k\in\{1,\dots,m\}$ and all $j\in\{1,\dots,r^\star(y_k)\}$, there exist polynomials which depend only on the states $\ket{y_k}$ such that $\xi_{k;j}=\mathrm{poly}\left(\epsilon,\frac1m\right)$ and
\begin{equation}
    \left|P(y_1,\dots, y_m \lvert\hat \rho)-\tilde P_{\bm\xi}(y_1,\dots, y_m \lvert\hat \rho)\right|\le\epsilon.
\end{equation}
\end{theorem}

\noindent We defer the proof of the theorem to the end of this section. This theorem implies that we can obtain arbitrarily good additive approximations of any individual output probability $P$ by computing $\tilde P_{\bm\xi}$ instead with the parameters $\xi_{k;j}$ going to $0$. In a practical computation however, we would be limited to working with numbers of polynomially many bits (in $m$), i.e.\ exponentially small values for $\xi_{k;j}$, and thus exponentially small values for $\epsilon$, i.e.\ $\epsilon=\mathcal O\left(\mathrm{exp}(-\mathrm{poly}\;m)\right)$. By Theorem~\ref{th:approx}, we have
\begin{equation}
    \sum_{\bm y=(y_1,\dots,y_m)}\left|P(y_1,\dots, y_m \lvert\hat \rho)-\tilde P_{\bm\xi}(y_1,\dots, y_m \lvert\hat \rho)\right|\le|\mathcal Y|\epsilon,
\end{equation}
where $\mathcal Y$ is the outcome space.
This implies that we can reproduce the output probability distribution $P$ over any exponentially sized outcome space up to any inverse-exponential error in total variation distance.

\medskip

\noindent Moreover, the scaling in Theorem~\ref{th:approx} is quite loose and the approximation works very well in practice as shown, for example, in Fig.~\ref{fig:simu}. To generate this figure, we used our protocol to construct the distributions $\tilde P_{\bm\xi}(y_1,\dots, y_m \lvert\hat \rho)$ (denoted $P_{\rm estimate}$ in the figure) that estimate Boson Sampling and Gaussian Boson Sampling probabilities. For both protocols, the exact probabilities are indicated by $P_{\rm exact}$. For the case of Boson Sampling the exact probability was calculated using the permanent \cite{10.1145/1993636.1993682} and for Gaussian Boson Sampling we use the hafnian as described in \cite{PhysRevLett.119.170501}. For the estimated probability we use the protocol presented in \cite{chabaud2020classical}, which is also based on hafnians. Each curve presented in Fig.~\ref{fig:simu} is based on 40 numerically calculated points, each for a randomly chosen interferometer and one randomly chosen output configuration with at most one detected photon per output mode. Because we have are most one photon in each output mode, we can simply set $\xi_{k;j} = \xi$ for all $k$ and $j$. Note that we present the multiplicative error $\abs{P_{\rm estimate} - P_{\rm exact}}/P_{\rm exact}$ to correct for potential large fluctuations in the exact probabilities. These results thus demonstrate the success of our method. The practical implementation also highlights that our method is highly inefficient as a computational tool, emphasising it as a worst case scenario and a theoretical resource identification tool. 

\begin{figure}
	\begin{center}
		\includegraphics[width=\columnwidth]{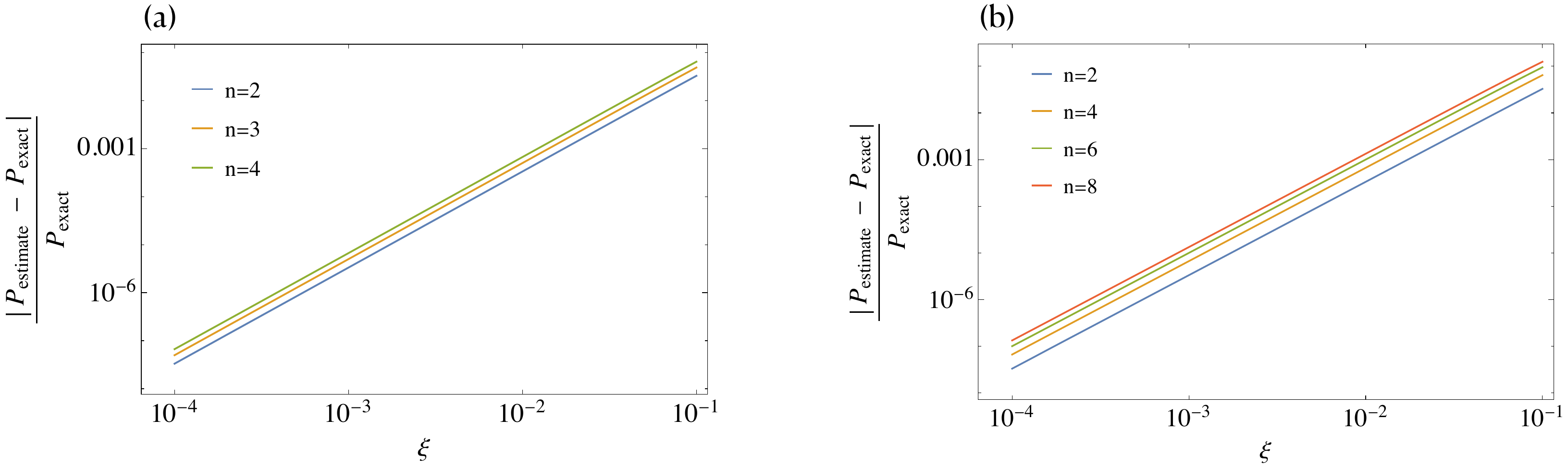}
    		\caption{Quality of probability estimates for different values of the approximation parameter $\xi$ and for random (a) Boson Sampling and (b) Gaussian Boson Sampling computations. Each curve is obtained from 40 numerically calculated points, each for a different randomly chosen output configuration of a randomly chosen interferometer.}
		\label{fig:simu}
	\end{center}
\end{figure}

\medskip 
 
\noindent The rest of this section is devoted to the proof of Theorem~\ref{th:approx}. At a high level, the proof proceeds as follows:
\begin{itemize}
    \item We first reduce the problem to proving an upper bound on the trace distance between the $m$-mode states $\bigotimes_{k=1}^m\ket{y_k}$ and $\bigotimes_{k=1}^m\ket{\tilde y_k}$ (see Eq.~(\ref{SM:dist1})), where $\ket{\tilde y_k}$ is the state obtained from $\ket{y_k}$ by replacing all $r^\star(y_k)$ photon additions $\hat a^\dag$ by attenuated photon additions $\hat a^\dag(\cosh\xi_{k;j})^{-\hat a^\dag\hat a}$, for some small $\xi_{k;j}\in\mathbb R$, for all $j\in\{1,\dots,r^\star(y_k)\}$ and all $k\in\{1,\dots,m\}$ (see Eq.~(\ref{SM:deftildeyk})).
    \item We further reduce the problem to proving an upper bound on the trace distance  between the single-mode states $\ket{y_k}$ and $\ket{\tilde y_k}$ for all $k\in\{1,\dots,m\}$ (see Eq.~(\ref{SM:disttensorprod})).
    \item This single-mode bound is obtained by induction, by first bounding the error introduced by replacing a single photon addition by an attenuated one (Lemma~\ref{lem:momentbound}), and then exploiting a recurrence relation for the error (see Eq.~(\ref{SM:recurrence})).
    \item Additional technicalities arise (such as Lemma~\ref{lem:TDnorm}) because some of the states considered (e.g., $\ket{\tilde y_k}$) are not normalised.
\end{itemize}

\begin{proof}
Replacing each photon subtraction in the third line of Eq.~(\ref{SM:outputprobaexpand}) by a renormalised map $\left(-\frac{\cosh^2 \xi}{\sinh \xi}\right)\mathcal E_\xi$ for some $\xi\in\mathbb R$ is equivalent to replacing each photon addition in the second line of Eq.~(\ref{SM:outputprobaexpand}) by its adjoint map
\begin{equation}\label{SM:Exiadj}
    \mathcal E^*_\xi(\ket\phi):=\tr_{\rm aux} \left\{\hat U_{\rm TMSS}(\xi)^\dag [\ket\phi\!\bra\phi\otimes \ket0\!\bra0_{\rm aux}] \hat U_{\rm TMSS}(\xi) [\hat{\mathds{1}}\otimes\ket1\!\bra1_{\rm aux}]\right\},
\end{equation}
and renormalising by $\left(-\frac{\cosh^2 \xi}{\sinh \xi}\right)$. For $\xi\in\mathbb R$, the adjoint map satisfies (see Eq.~(\ref{SM:devExi}))
\begin{equation}
    \mathcal E^*_\xi(\ket\phi)=\left(-\frac{\sinh\xi}{\cosh^2\xi}\right)\hat a^\dag(\cosh\xi)^{-\hat a^\dag\hat a}\ket\phi.
\end{equation}
We define, for all $k\in\{1,\dots,m\}$, the unnormalised states
\begin{equation}\label{SM:deftildeyk}
    \ket{\tilde y_k}:=\frac1{\sqrt{{\cal N}_k}} \left[\prod_{j=1}^{r^\star(y_k)} \hat D(\beta_{k;j})\left[\hat a^{\dag}_k(\cosh\xi_{k;j})^{-\hat a^\dag\hat a}\right]\hat D^{\dag}(\beta_{k;j})\right]\hat S_k\ket{\alpha_k},
\end{equation}
where $\xi_{k;j}\in\mathbb R$ for all $j\in\{1,\dots,r^\star(y_k)\}$. These states are obtained by replacing each photon addition $\hat a^\dag$ in Eq.~(\ref{SM:defyk}) by $\hat a^\dag(\cosh\xi_{k;j})^{-\hat a^\dag\hat a}$ for some $\xi_{k;j}\in\mathbb R$ (i.e.\ by the action of a renormalised map $\left(-\frac{\cosh^2 \xi_{k;j}}{\sinh \xi_{k;j}}\right)\mathcal E^*_\xi$).
With Eq.~(\ref{SM:tildePxi}), we thus obtain
\begin{equation}\label{SM:tildePxiwithtildeyk}
    \tilde P_{\bm\xi}(y_1,\dots, y_m \lvert\hat \rho)=\tr\left[\hat\rho\bigotimes_{k=1}^m\ket{\tilde y_k}\!\bra{\tilde y_k}\right].
\end{equation}
Hence, 
\begin{equation}\label{SM:dist1}
    \begin{aligned}
    \left|P(y_1,\dots, y_m \lvert\hat \rho)-\tilde P_{\bm\xi}(y_1,\dots, y_m \lvert\hat \rho)\right|&=\left|\tr\left[\hat\rho\bigotimes_{k=1}^m\ket{y_k}\!\bra{y_k}\right]-\tr\left[\hat\rho\bigotimes_{k=1}^m\ket{\tilde y_k}\!\bra{\tilde y_k}\right]\right|\\
    &\le\tr\left[\hat\rho\left|\bigotimes_{k=1}^m\ket{y_k}\!\bra{y_k}-\bigotimes_{k=1}^m\ket{\tilde y_k}\!\bra{\tilde y_k}\right|\right]\\
    &\le\left\|\bigotimes_{k=1}^m\ket{y_k}\!\bra{y_k}-\bigotimes_{k=1}^m\ket{\tilde y_k}\!\bra{\tilde y_k}\right\|_1\\
    &=2D\left(\bigotimes_{k=1}^m\ket{y_k},\bigotimes_{k=1}^m\ket{\tilde y_k}\right),
    \end{aligned}
\end{equation}
where we used the sub-multiplicativity of the $1$-norm and $\|\hat\rho\|_1=\tr(\hat\rho)=1$ in the third line.

It is thus sufficient to prove an upper bound $\frac\epsilon2$ on the trace distance between the states $\bigotimes_{k=1}^m\ket{y_k}$ and $\bigotimes_{k=1}^m\ket{\tilde y_k}$ in order to prove the theorem. To achieve this, we introduce the following intermediate results:

\begin{lemma}\label{lem:TDnorm}
Let $\ket\psi$ be a normalised pure state and let $|\tilde\psi\rangle$ be a pure state, possibly unnormalised. Then, writing $\left\||\tilde\psi\rangle\right\|^2=\big\langle\tilde\psi|\tilde\psi\big\rangle$,
\begin{equation}
    D\left(\ket\psi,|\tilde\psi\rangle\right)=\sqrt{\frac14\left(1+\left\||\tilde\psi\rangle\right\|^2\right)^2-\left|\bra\psi\tilde\psi\rangle\right|^2}.
\end{equation}
In particular, $\left\||\tilde\psi\rangle\right\|^2\le1$ implies $D\left(\ket\psi,|\tilde\psi\rangle\right)\le\sqrt{1-\left|\bra\psi\tilde\psi\rangle\right|^2}$.
\end{lemma}

\textit{Proof of Lemma~\ref{lem:TDnorm}.} We write $|\tilde\psi\rangle=\left\||\tilde\psi\rangle\right\|\left(\cos\theta\ket\psi+\sin\theta\ket{\psi^\perp}\right)$ for some $\theta\in[0,2\pi]$ and some normalised pure state $\ket{\psi^\perp}$ orthogonal to $\ket\psi$. We have $\left|\bra\psi\tilde\psi\rangle\right|^2=\cos^2\theta$, and
\begin{equation}
    \ket\psi\!\bra\psi-|\tilde\psi\rangle\langle\tilde\psi|=\begin{pmatrix}1-\left\||\tilde\psi\rangle\right\|^2\cos^2\theta&-\left\||\tilde\psi\rangle\right\|^2\sin\theta\cos\theta\\-\left\||\tilde\psi\rangle\right\|^2\sin\theta\cos\theta&-\left\||\tilde\psi\rangle\right\|^2\sin^2\theta\end{pmatrix},
\end{equation}
in the $\left(\ket\psi,\ket{\psi^\perp}\right)$ basis. The modulus of the eigenvalues of this matrix are given by
\begin{equation}
    |\lambda_\pm|=\frac12\left[\sqrt{\left(1-\left\||\tilde\psi\rangle\right\|^2\right)^2+4\left\||\tilde\psi\rangle\right\|^2\sin^2\theta}\pm\left(1-\left\||\tilde\psi\rangle\right\|^2\right)\right].
\end{equation}
Hence,
\begin{equation}
    \begin{aligned}
    D&=\frac12\left\|\ket\psi\!\bra\psi-|\tilde\psi\rangle\langle\tilde\psi|\right\|_1\\
    &=\frac12\left(|\lambda_+|+|\lambda_-|\right)\\
    &=\frac12\sqrt{\left(1-\left\||\tilde\psi\rangle\right\|^2\right)^2+4\left\||\tilde\psi\rangle\right\|^2\sin^2\theta}\\
    &=\sqrt{\frac14\left(1+\left\||\tilde\psi\rangle\right\|^2\right)^2+4\left\||\tilde\psi\rangle\right\|^2\cos^2\theta}\\
    &=\sqrt{\frac14\left(1+\left\||\tilde\psi\rangle\right\|^2\right)^2-\left|\bra\psi\tilde\psi\rangle\right|^2}.
    \end{aligned}
\end{equation}
\qed
\begin{lemma}\label{lem:momentbound}
Let $\ket\phi$ be a normalised pure state with bounded moments $\bra\phi\hat a^\dag\hat a\ket\phi=E\ge0$ and $\bra\phi(\hat a^\dag\hat a)^2\ket\phi=M\ge0$. Let $\xi\in\mathbb R$, and let
\begin{equation}
    \ket\psi:=\frac{\hat a^\dag\ket\phi}{\|\hat a^\dag\ket\phi\|}\quad\text{and}\quad\ket{\tilde\psi}:=\frac{\hat a^\dag(\cosh\xi)^{-\hat a^\dag\hat a}\ket\phi}{\|\hat a^\dag\ket\phi\|}.
\end{equation}
Then, $\left\||\tilde\psi\rangle\right\|^2\le1$ and
\begin{equation}
    D\left(\ket\psi,\ket{\tilde\psi}\right)\le C\sqrt\xi.
\end{equation}
where $C=\left(\frac{8(E+M)}{E+1}\right)^{1/4}$.
\end{lemma}

\textit{Proof of Lemma~\ref{lem:momentbound}.} Let us denote $\ket\phi=\sum_{k\ge0}\phi_k\ket k$ in Fock basis. We have
\begin{equation}
    \begin{aligned}
        \left\|\hat a^\dag\ket\phi\right\|^2\left\||\tilde\psi\rangle\right\|^2&=\left\|\hat a^\dag(\cosh\xi)^{-\hat a^\dag\hat a}\ket\phi\right\|\\
        &=\sum_{k\ge0}|\phi_k|^2(k+1)(\cosh\xi)^{-2k}\\
        &\le\sum_{k\ge0}|\phi_k|^2(k+1)\\
        &=\left\|\hat a^\dag\ket\phi\right\|^2,
    \end{aligned}
\end{equation}
so $\left\||\tilde\psi\rangle\right\|^2\le1$.
For all $N>0$, we have $\sum_{k>N}|\phi_k|^2\le\frac1N\sum_{k>N}k|\phi_k|^2\le\frac EN$ and similarly $\sum_{k>N}k|\phi_k|^2\le\frac1N\sum_{k>N}k^2|\phi_k|^2\le\frac MN$, so
\begin{equation}\label{SM:tailbound}
    \sum_{k>N}(k+1)|\phi_k|^2\le\frac1N(E+M).
\end{equation}
Moreover,
\begin{equation}
    \begin{aligned}
        D\left(\ket\psi,\ket{\tilde\psi}\right)^2&\le1-\left|\bra\psi\tilde\psi\rangle\right|^2\\
        &=1-\left(\frac{\bra\phi\hat a\hat a^\dag(\cosh\xi)^{-\hat a^\dag\hat a}\ket\phi}{\bra\phi\hat a\hat a^\dag\ket\phi}\right)^2\\
        &=1-\left(\frac{\sum_{k\ge0}(k+1)(\cosh\xi)^{-k}|\phi_k|^2}{E+1}\right)^2\\
        &=1-\left(\frac{\sum_{k\ge0}(k+1)(\cosh\xi)^{-k}|\phi_k|^2}{E+1}\right)^2\\
        &=1-\left(1-\frac{\sum_{k\ge0}(k+1)[1-(\cosh\xi)^{-k}]|\phi_k|^2}{E+1}\right)^2,
    \end{aligned}
\end{equation}
where we used $\left\||\tilde\psi\rangle\right\|^2\le1$ and Lemma~\ref{lem:TDnorm} in the first line.
Let $N>0$, we have
\begin{equation}
    \begin{aligned}
        \sum_{k\ge0}(k+1)[1-(\cosh\xi)^{-k}]|\phi_k|^2&=\sum_{k\le N}(k+1)[1-(\cosh\xi)^{-k}]|\phi_k|^2+\sum_{k>N}(k+1)[1-(\cosh\xi)^{-k}]|\phi_k|^2\\
        &\le[1-(\cosh\xi)^{-N}]\sum_{k\le N}(k+1)|\phi_k|^2+\sum_{k>N}(k+1)|\phi_k|^2\\
        &\le[1-(\cosh\xi)^{-N}]\sum_{k\ge0}(k+1)|\phi_k|^2+\sum_{k>N}(k+1)|\phi_k|^2\\
        &\le[1-(\cosh\xi)^{-N}](E+1)+\frac1N(E+M),
    \end{aligned}
\end{equation}
where we used Eq.~(\ref{SM:tailbound}) in the last line. Putting things together, we obtain
\begin{equation}
    \begin{aligned}
        D\left(\ket\psi,\ket{\tilde\psi}\right)^2&\le1-\left[(\cosh\xi)^{-N}-\frac1N\left(\frac{E+M}{E+1}\right)\right]^2\\
        &\le1-(\cosh\xi)^{-2N}+\frac2N\left(\frac{E+M}{E+1}\right)\\
        &\le N\xi^2+\frac2N\left(\frac{E+M}{E+1}\right).
    \end{aligned}
\end{equation}
Choosing $N=\frac1\xi\sqrt{\frac{2(E+M)}{E+1}}$ and taking the square root completes the proof of the lemma.\quad\qedsymbol

\medskip

\noindent We now combine these results to bound the trace distance between the states $\ket{y_k}$ and $\ket{\tilde y_k}$ by induction.
Let $r\ge0$, let $\ket{\phi_0}$ be a normalised pure state, let $\gamma_1,\dots,\gamma_r\in\mathbb C$, and let $\xi_1,\dots,\xi_r\in\mathbb R$. For all $j\in\{1,\dots,r\}$, let
\begin{equation}
    \ket{\phi_j}:=\frac1{\sqrt{\mathcal N(j)}}\hat D(\gamma_j)\hat a^\dag\hat D(\gamma_{j-1})\dots\hat D(\gamma_1)\hat a^\dag\ket{\phi_0},
\end{equation}
where $\mathcal N(j)=\|\hat D(\gamma_j)\hat a^\dag\hat D(\gamma_{j-1})\dots\hat D(\gamma_1)\hat a^\dag\ket{\phi_0}\|^2$ is a normalisation factor, and let $\ket{\tilde\phi_0}=\ket{\phi_0}$ and
\begin{equation}
    \ket{\tilde\phi_j}:=\frac1{\sqrt{\mathcal N(j)}}\hat D(\gamma_j)\hat a^\dag(\cosh\xi_j)^{-\hat a^\dag\hat a}\hat D(\gamma_{j-1})\dots\hat D(\gamma_1)a^\dag(\cosh\xi_1)^{-\hat a^\dag\hat a}\ket{\phi_0}.
\end{equation}
For all $j\in\{1,\dots,r\}$, we write $E_j:=\bra{\phi_j}\hat a^\dag\hat a\ket{\phi_j}$ and $M_j:=\bra{\phi_j}(\hat a^\dag\hat a)^2\ket{\phi_j}$. We have
\begin{equation}
    \begin{aligned}
        D\left(\ket{\phi_j},\ket{\tilde\phi_j}\right)&=\frac12\left\|\ket{\phi_j}\!\bra{\phi_j}-\ket{\tilde\phi_j}\!\bra{\tilde\phi_j}\right\|_1\\
        &=\frac12\frac{\mathcal N(j-1)}{\mathcal N(j)}\left\|\hat D(\gamma_j)\hat a^\dag\ket{\phi_{j-1}}\!\bra{\phi_{j-1}}\hat a\hat D(\gamma_j)^\dag-\hat D(\gamma_j)\hat a^\dag(\cosh\xi_j)^{-\hat a^\dag\hat a}\ket{\tilde\phi_{j-1}}\!\bra{\tilde\phi_{j-1}}(\cosh\xi_j)^{-\hat a^\dag\hat a}\hat a\hat D(\gamma_j)^\dag\right\|_1\\
        &=\frac12\frac{\mathcal N(j-1)}{\mathcal N(j)}\left\|\hat a^\dag\ket{\phi_{j-1}}\!\bra{\phi_{j-1}}\hat a-\hat a^\dag(\cosh\xi_j)^{-\hat a^\dag\hat a}\ket{\tilde\phi_{j-1}}\!\bra{\tilde\phi_{j-1}}(\cosh\xi_j)^{-\hat a^\dag\hat a}\hat a\right\|_1\\
        &=\frac12\frac{\mathcal N(j-1)}{\mathcal N(j)}\Big\|\hat a^\dag\ket{\phi_{j-1}}\!\bra{\phi_{j-1}}\hat a-a^\dag(\cosh\xi_j)^{-\hat a^\dag\hat a}\ket{\phi_{j-1}}\!\bra{\phi_{j-1}}(\cosh\xi_j)^{-\hat a^\dag\hat a}\hat a\\
        &\quad\quad\quad\quad\quad\quad+a^\dag(\cosh\xi_j)^{-\hat a^\dag\hat a}\ket{\phi_{j-1}}\!\bra{\phi_{j-1}}(\cosh\xi_j)^{-\hat a^\dag\hat a}\hat a-\hat a^\dag(\cosh\xi_j)^{-\hat a^\dag\hat a}\ket{\tilde\phi_{j-1}}\!\bra{\tilde\phi_{j-1}}(\cosh\xi_j)^{-\hat a^\dag\hat a}\hat a\Big\|_1\\
        &\le\frac12\frac{\mathcal N(j-1)}{\mathcal N(j)}\left\|\hat a^\dag\ket{\phi_{j-1}}\!\bra{\phi_{j-1}}\hat a-a^\dag(\cosh\xi_j)^{-\hat a^\dag\hat a}\ket{\phi_{j-1}}\!\bra{\phi_{j-1}}(\cosh\xi_j)^{-\hat a^\dag\hat a}\hat a\right\|_1\\
        &\quad+\frac12\frac{\mathcal N(j-1)}{\mathcal N(j)}\left\|a^\dag(\cosh\xi_j)^{-\hat a^\dag\hat a}\ket{\phi_{j-1}}\!\bra{\phi_{j-1}}(\cosh\xi_j)^{-\hat a^\dag\hat a}\hat a-\hat a^\dag(\cosh\xi_j)^{-\hat a^\dag\hat a}\ket{\tilde\phi_{j-1}}\!\bra{\tilde\phi_{j-1}}(\cosh\xi_j)^{-\hat a^\dag\hat a}\hat a\right\|_1,
    \end{aligned}
\end{equation}
where we used the triangle inequality in the last line.
Introducing the notations $\ket{\psi_j}$ and $\ket{\tilde\psi_j}$ as in Lemma~\ref{lem:momentbound}, and noting that $\frac{\mathcal N(j)}{\mathcal N(j-1)}=\left\|\hat a^\dag\ket{\phi_{j-1}}\right\|^2=E_{j-1}+1$, we obtain the following recurrence relation:
\begin{equation}\label{SM:recurrence}
    \begin{aligned}
        D\left(\ket{\phi_j},\ket{\tilde\phi_j}\right)&\le D\left(\ket{\tilde\psi_j},\ket{\psi_j}\right)+\frac1{2(E_{j-1}+1)}\left\|a^\dag(\cosh\xi_j)^{-\hat a^\dag\hat a}\left(\ket{\phi_{j-1}}\!\bra{\phi_{j-1}}-\ket{\tilde\phi_{j-1}}\!\bra{\tilde\phi_{j-1}}\right)(\cosh\xi_j)^{-\hat a^\dag\hat a}\hat a\right\|_1\\
        &=D\left(\ket{\tilde\psi_j},\ket{\psi_j}\right)+\frac1{2(E_{j-1}+1)}\left(\frac{\cosh^2\xi_j}{\sinh\xi_j}\right)^2\left\|\mathcal E^*_{\xi_j}\left(\ket{\phi_{j-1}}\!\bra{\phi_{j-1}}-\ket{\tilde\phi_{j-1}}\!\bra{\tilde\phi_{j-1}}\right)\right\|_1\\
        &\le D\left(\ket{\tilde\psi_j},\ket{\psi_j}\right)+\frac1{E_{j-1}+1}\left(\frac{\cosh^2\xi_j}{\sinh\xi_j}\right)^2D\left(\ket{\tilde\phi_{j-1}},\ket{\phi_{j-1}}\right)\\
        &\le C_j\sqrt{\xi_j}+\frac1{E_{j-1}+1}\left(\frac{\cosh^2\xi_j}{\sinh\xi_j}\right)^2D\left(\ket{\tilde\phi_{j-1}},\ket{\phi_{j-1}}\right)\\
        &\le C_j\sqrt{\xi_j}+\frac2{\xi_j^2(E_{j-1}+1)}D\left(\ket{\tilde\phi_{j-1}},\ket{\phi_{j-1}}\right),
    \end{aligned}
\end{equation}
where we used the definition of the CPTD map $\mathcal E^*_\xi$ in the second line, the fact that the trace distance is non-increasing under quantum operations in the third line, Lemma~\ref{lem:momentbound} in the fourth line (the constant $C_j$ is the one corresponding to the state $\ket{\phi_j}$ in Lemma~\ref{lem:momentbound}), and the bound $\frac{\cosh^2x}{\sinh x}\le\frac2x$ in the last line (valid for $x$ small enough). We have $\ket{\tilde\phi_0}=\ket{\phi_0}$, so a simple induction using the equation above yields
\begin{equation}
    D\left(\ket{\phi_r},\ket{\tilde\phi_r}\right)\le\sum_{j=1}^r\frac{K_j\sqrt{\xi_j}}{\prod_{i=j+1}^r\xi_i^2},
\end{equation}
where we have defined
\begin{equation}
    K_j:=C_j\prod_{i=j+1}^r\frac2{E_{i-1}+1}.
\end{equation}
For all $k\in\{1,\dots,m\}$, the states $\ket{y_k}$ and $\ket{\tilde y_k}$ are equal to the states $\ket{\phi_{r^\star(y_k)}}$ and $\ket{\tilde\phi_{r^\star(y_k)}}$, respectively, when setting $\gamma_{r^\star(y_k)}=\beta_{k;r^\star(y_k)}$, $\gamma_j=\beta_{k;j-1}-\beta_{k;j}$ for all $j\in\{1,\dots,r^\star(y_k)-1\}$, $\gamma_0=\beta_{k;0}$, $\ket{\phi_0}=\ket{\tilde\phi_0}=\hat S_k\ket{\alpha_k}$, and $\xi_j=\xi_{k;j}$  for all $j\in\{1,\dots,r^\star(y_k)\}$. For all $k\in\{1,\dots,m\}$ and all $j\in\{1,\dots,r^\star(y_k)\}$, we define the intermediate states leading to $\ket{y_k}$ (see Eq.~(\ref{SM:defyk})) as
\begin{equation}
    \ket{\phi_{k;j}}:=\frac1{\sqrt{\mathcal N_{k;j}}}\left[\prod_{i=1}^j \hat D(\beta_{k;i}) \hat a^{\dag}_k \hat D^{\dag}(\beta_{k;i})\right]\hat S_k\ket{\alpha_k},
\end{equation}
where $\mathcal N_{k;j}$ is a normalisation factor.
We thus obtain,
\begin{equation}\label{SM:bounddistyk0}
    D\left(\ket{y_k},\ket{\tilde y_k}\right)\le\sum_{j=1}^{r^\star(y_k)}\frac{K_{k;j}\sqrt{\xi_{k;j}}}{\prod_{i=j+1}^{r^\star(y_k)}\xi_{k;i}^2},
\end{equation}
where
\begin{equation}
    K_{k;j}:=\left(\frac{8(E_{k;j}+M_{k;j})}{E_{k;j}+1}\right)^{1/4}\prod_{i=j+1}^{r^\star(y_k)}\frac2{E_{k;i-1}+1},
\end{equation}
with $E_{k;j}:=\bra{\phi_{k;j}}\hat a^\dag\hat a\ket{\phi_{k;j}}$ and $M_{k;j}:=\bra{\phi_{k;j}}(\hat a^\dag\hat a)^2\ket{\phi_{k;j}}$. We now set recursively
\begin{equation}\label{xikjeps}
    \xi_{k;j}=\frac{\epsilon^2}{4m^2}\left(\frac{\prod_{i=j+1}^{r^\star(y_k)}\xi_{k;i}^2}{r^\star(y_k) K_{k;j}}\right)^2,
\end{equation}
for $j=r^\star(y_k),\dots,1$. A simple induction shows that $\xi_{k;j}$ is a polynomial in $\epsilon$ and $\frac1m$, with the coefficients of this polynomial being entirely determined by the state $\ket{y_k}$ (more precisely by $r^\star(y_k)$ and by the values of $E_{k;i}$ and $M_{k;i}$ for all $i\in\{j,\dots,r^\star(y_k)\}$). Moreover, with this choice, Eq.~(\ref{SM:bounddistyk0}) directly yields
\begin{equation}\label{SM:bounddistyk}
    D\left(\ket{y_k},\ket{\tilde y_k}\right)\le \frac\epsilon{2m},
\end{equation}
for all $k\in\{1,\dots,m\}$.
Note that with a simple induction, the first part of Lemma~\ref{lem:momentbound} shows that $\left\||\tilde y_k\rangle\right\|^2\le1$, for all $k\in\{1,\dots,m\}$. As a consequence,
\begin{equation}\label{SM:disttensorprod}
    \begin{aligned}
        D\left(\bigotimes_{k=1}^m\ket{y_k},\bigotimes_{k=1}^m\ket{\tilde y_k}\right)&=\frac12\left\|\bigotimes_{k=1}^m\ket{y_k}\!\bra{y_k}-\bigotimes_{k=1}^m\ket{\tilde y_k}\!\bra{\tilde y_k}\right\|_1\\
        &=\frac12\left\|\sum_{l=1}^m\bigotimes_{k=1}^{l-1}\ket{y_k}\!\bra{y_k}\otimes(\ket{y_l}\!\bra{y_l}-\ket{\tilde y_l}\!\bra{\tilde y_l})\otimes\bigotimes_{k=l+1}^m\ket{\tilde y_k}\!\bra{\tilde y_k}\right\|_1\\
        &\le\frac12\sum_{l=1}^m\left\|\bigotimes_{k=1}^{l-1}\ket{y_k}\!\bra{y_k}\otimes(\ket{y_l}\!\bra{y_l}-\ket{\tilde y_l}\!\bra{\tilde y_l})\otimes\bigotimes_{k=l+1}^m\ket{\tilde y_k}\!\bra{\tilde y_k}\right\|_1\\
        &\le\frac12\sum_{l=1}^m\left(\prod_{k=l+1}^m\|\ket{\tilde y_k}\!\bra{\tilde y_k}\|_1\right)\left\|\ket{y_l}\!\bra{y_l}-\ket{\tilde y_l}\!\bra{\tilde y_l}\right\|_1\\
        &\le\frac12\sum_{l=1}^m\left\|\ket{y_l}\!\bra{y_l}-\ket{\tilde y_l}\!\bra{\tilde y_l}\right\|_1\\
        &=\sum_{l=1}^mD\left(\ket{y_l},\ket{\tilde y_l}\right),
    \end{aligned}
\end{equation}
where we used the definition of the trace distance in the first and last lines, a telescopic sum in the second line, the triangle inequality in the third line, the submultplicativity of the $1$-norm in the fourth line together with the fact that the states $\ket{y_k}$ are normalised, and $\left\|\ket{\tilde y_l}\!\bra{\tilde y_l}\right\|_1=\left\||\tilde y_k\rangle\right\|^2\le1$ in the fifth line.
Finally, combining Eqs.~(\ref{SM:bounddistyk}) and~(\ref{SM:disttensorprod}) yields
\begin{equation}\label{SM:boundDtensor}
    \begin{aligned}
        D\left(\bigotimes_{k=1}^m\ket{y_k},\bigotimes_{k=1}^m\ket{\tilde y_k}\right)&\le\sum_{l=1}^m D\left(\ket{y_l},\ket{\tilde y_l}\right)\\
        &\le\frac\epsilon2.
    \end{aligned}
\end{equation}
This implies
\begin{equation}
    \left|P(y_1,\dots, y_m \lvert\hat \rho)-\tilde P_{\bm\xi}(y_1,\dots, y_m \lvert\hat \rho)\right|\le\epsilon,
\end{equation}
by Eq.~(\ref{SM:dist1}).
\end{proof}

Note that in Eq.~(3) in the main text we used for brevity the normalisation factor $\frac1{\xi^2}$ in the definition of the approximate probabilities rather than $\left(\frac{\cosh^2\xi}{\sinh\xi}\right)^2$.
In the above proof, we showed $\||\tilde y_k\rangle\|^2\le1$ for all $k\in\{1,\dots,m\}$, so with Eq.~(\ref{SM:tildePxiwithtildeyk}) we have $\tilde P_{\bm\xi}(y_1,\dots, y_m \lvert\hat \rho)\le1$. Defining
\begin{equation}
    P_\text{estimate}(y_1,\dots, y_m \lvert\hat \rho):=\prod_{k=1}^m\!\prod_{j=1}^{r^\star(y_k)}\left(\frac{\sinh\xi_{k;j}}{\xi_{k;j}\cosh^2\xi_{k;j}}\right)^2\tilde P_{\bm\xi}(y_1,\dots, y_m \lvert\hat \rho),
\end{equation}
as in the main text, we obtain
\begin{equation}
    \begin{aligned}
        \left|\tilde P_{\bm\xi}(y_1,\dots, y_m \lvert\hat \rho)-P_\text{estimate}(y_1,\dots, y_m \lvert\hat \rho)\right|&=\left|1-\prod_{k=1}^m\!\prod_{j=1}^{r^\star(y_k)}\left(\frac{\sinh\xi_{k;j}}{\xi_{k;j}\cosh^2\xi_{k;j}}\right)^2\right|\tilde P_{\bm\xi}(y_1,\dots, y_m \lvert\hat\rho)\\
        &\le\left|1-\prod_{k=1}^m\!\prod_{j=1}^{r^\star(y_k)}\left(\frac{\sinh\xi_{k;j}}{\xi_{k;j}\cosh^2\xi_{k;j}}\right)^2\right|.
    \end{aligned}
\end{equation}
We have $\frac{\sinh\xi}{\xi\cosh^2\xi}=1-\mathcal O(\xi^2)$ when $\xi$ goes to $0$. Moreover, with Eq.~(\ref{xikjeps}), we have $\xi_{k;j}=o(\frac\epsilon m)$ for all $k\in\{1,\dots,m\}$ and all $j\in\{1,\dots,r^\star(y_k)\}$, so
\begin{equation}
    \begin{aligned}
        \left|\tilde P_{\bm\xi}(y_1,\dots, y_m \lvert\hat \rho)-P_\text{estimate}(y_1,\dots, y_m \lvert\hat \rho)\right|&\le\left|1-\left[1-o\left(\frac{\epsilon^2} {m^2}\right)\right]^{\sum_{k=1}^mr^\star(y_k)}\right|\\
        &\le\mathcal O(\epsilon).
    \end{aligned}
\end{equation}
Together with Theorem~\ref{th:approx} and the triangle inequality, this justifies Eq.~(3) in the main text.

\section{Coherent state sampling for marginal probabilities}
\label{app:marginals}

Let $p\in\{1,\dots,m\}$. With the notations of the previous section, sampling from the $p$-marginal of the bosonic sampling setup amounts to sampling from
\begin{equation}
    P(y_1,\dots, y_p \lvert \hat \rho) = \tr \left[\hat \rho \bigotimes_{k=1}^p\ket{y_k}\!\bra{y_k}\otimes\bigotimes_{k=p+1}^m\hat{\mathds1}\right].
\end{equation}
Note that we consider the first $p$ coordinates without loss of generality, up to a permutation. We define as in the previous section
\begin{equation}
    \tilde P_{\bm\xi}(y_1,\dots, y_p \lvert \hat \rho) := \tr \left[\hat \rho \bigotimes_{k=1}^p\ket{\tilde y_k}\!\bra{\tilde y_k}\otimes\bigotimes_{k=p+1}^m\hat{\mathds1}\right].
\end{equation}
where for all $k\in\{1,\dots,p\}$, $\ket{\tilde y_k}$ is the state obtained from $\ket{y_k}$ by replacing all $r^\star(y_k)$ photon additions $\hat a^\dag$ by attenuated photon additions $\hat a^\dag(\cosh\xi_{k;j})^{-\hat a^\dag\hat a}$, for some $\xi_{k;j}\in\mathbb R$ (see Eqs.~(\ref{SM:tildePxiwithtildeyk}) and~(\ref{SM:deftildeyk})).

By Theorem~\ref{th:approx} (and Eq.~(\ref{SM:boundDtensor}) in particular), for all $k\in\{1,\dots,p\}$ and all $j\in\{1,\dots,r^\star(y_k)\}$, there exist polynomials which depend only on the states $\ket{y_k}$ such that $\xi_{k;j}=\mathrm{poly}\left(\epsilon,\frac1m\right)$ and
\begin{equation}\label{SM:boundmarginal}
    \left|P(y_1,\dots, y_p \lvert\hat \rho)-\tilde P_{\bm\xi}(y_1,\dots, y_p \lvert\hat \rho)\right|\le\epsilon.
\end{equation}
Similar to Eq.~(\ref{SM:tildePxi}), we also have
\begin{equation}\label{eq:marginalreduction}
    \tilde P_{\bm\xi}(y_1,\dots, y_p\lvert \hat \rho)=\frac1{\cal N}\prod_{k=1}^p\prod_{j=1}^{r^\star(y_k)}\left(\frac{\cosh^2(\xi_{k;j})}{\sinh (\xi_{k;j})}\right)^2\tr\left[ \hat \rho_{\rm total}^{(p)} \left(\bigotimes_{k=1}^p\ket{\alpha_k}\!\bra{\alpha_k}\otimes\bigotimes_{k=p+1}^m\hat{\mathds1} \otimes\ket0\!\bra0^{\otimes\sum_{k=1}^p r^\star(y_k)}\right) \right],
\end{equation}
where $r^\star(y_1),\dots,r^\star(y_p)$ are the stellar ranks of the states $\ket{y_1},\dots,\ket{y_p}$, respectively, and where the state $\hat \rho_{\rm total}^{(p)}$ is defined by
\begin{equation}\label{eq:marginaltotalstate}
   \hat \rho_{\rm total}^{(p)}:= (\hat{\bm S}^{\dag} \otimes \hat{\mathds{1}}_{\rm aux})\hat {\cal U}^{\dag} \left(\hat \rho \otimes\ket{1}\!\bra{1}^{\otimes\sum_{k=1}^p r^\star(y_k)}\right) \hat {\cal U} (\hat{\bm S}\otimes\hat{\mathds{1}}_{\rm aux}),
\end{equation}
where $\hat{\bm S}:=\bigotimes_{k=1}^p\hat S_k$ and where $\hat {\cal U}$ is given by $\hat {\cal U} := \bigotimes_{k=1}^p\prod_{n=1}^{r^\star(y_k)} \hat D(\beta_{k;j}) \hat U^{\dag}(\xi_{k;j}) \hat D^{\dag}(\beta_{k;j})$.

Hence, any $p$-marginal probability of a bosonic sampling setup with input state $\hat\rho$ can be approximated up to $\epsilon$-additive precision by a $p$-marginal probability of an associated coherent state sampler, which takes as input the state $\hat\rho^{(p)}_{\text{total}}$, satisfying $r^\star(\hat\rho^{(p)})=r^\star(\hat\rho)+\sum_{k=1}^pr^\star(y_k)$. As in the previous section, we can pick $\epsilon=\mathcal O\left(\text{exp}(-\text{poly }m)\right)$ in Eq.~(\ref{SM:boundmarginal}) with only a polynomial time computational overhead in $m$.

\section{Strong simulation of bosonic computations}
\label{app:strong}

In this section, we prove a generalized version of Theorem~1 from the main text. We first recall the definition of strong simulation:

\begin{definition}[Strong simulation]
Let $P$ be a probability distribution (density). Strong simulation of $P$ refers to the computational task of evaluating $P$ or any marginal of $P$ at any given outcome.
\end{definition}

\noindent Importantly, under mild assumptions, the ability to strongly simulate a probability distribution (density) implies the ability to exactly sample from this distribution (density), also referred to as weak simulation~\cite{terhal2002classical,pashayan2015estimating}.

We also introduce a notion of approximate strong simulation:

\begin{definition}[$\epsilon$-approximate strong simulation]
Let $P$ be a probability distribution (density). For $\epsilon>0$, $\epsilon$-approximate strong simulation refers to the computational task of strongly simulating a probability distribution $Q$ such that $\|P-Q\|_\text{tvd}\le\epsilon$, where $\|\cdot\|_\text{tvd}$ is the total variation distance.
\end{definition}

\noindent As a direct consequence of this definition, the ability to $\epsilon$-approximately strongly simulate a probability distribution (density) implies the ability to approximately sample from this distribution (density), which is the relevant computational task in experimental demonstrations of quantum computational advantage~\cite{harrow2017quantum}.

Note that any bound on the total variation distance between two probability distributions also holds for the total variation distance between any of their corresponding marginals: let $P$ and $Q$ be two probability distributions over some sample space $\mathcal X^m$, and let $n\in\{1,\dots,m\}$. Writing $P_n$ and $Q_n$ the marginal distributions over, e.g., the first $n$ coordinates, we have:
\begin{equation}
    \begin{aligned}
        \|P_n-Q_n\|_\text{tvd}&=\frac12\sum_{x_1,\dots,x_n\in\mathcal X }\left|P_n(x_1,\dots,x_n)-Q_n(x_1,\dots,x_n)\right|\\
        &=\frac12\sum_{x_1,\dots,x_n\in\mathcal X}\left|\sum_{x_{n+1},\dots,x_m\in\mathcal X}\left[P(x_1,\dots,x_n,x_{n+1},\dots,x_m)-Q_(x_1,\dots,x_n,x_{n+1},\dots,x_m)\right]\right|\\
        &\le\frac12\sum_{x_1,\dots,x_m\in\mathcal X}\left|P(x_1,\dots,x_m)-Q_(x_1,\dots,x_m)\right|\\
        &=\|P-Q\|_\text{tvd},
    \end{aligned}
\end{equation}
where we used the triangle inequality in the third line.

We consider an $m$-mode bosonic computation over an exponentially large outcome space $\mathcal Y=\mathcal Y_1\times\dots\times\mathcal Y_m$ which samples an outcome $\bm y=(y_1,\dots,y_m)\in\mathcal Y$ using $m$ local detectors from the probability density
\begin{equation}
    \tr\left[\hat\rho\bigotimes_{k=1}^m\hat P_{k;y_k}\right],
\end{equation}
where $\hat\rho$ is an $m$-mode quantum state and where $\hat P_{k;y_k}$ are rank-one projectors without loss of generality, such that $\hat P_{k;y_k}=\ket{y_k}\!\bra{y_k}$.

The state $\hat\rho$ can be any multimode mixed state, but in a typical sampling setup it is generated by a series of (possibly noisy) few-mode gates applied to single-mode input states. We therefore assume that a purification $\hat\rho=\tr_{\mathcal E}(\ket{\bm\psi}\!\bra{\bm\psi})$ is known, where $\mathcal E$ is an $m$-mode purification Hilbert space. 

In most bosonic computations the state $\ket{\bm\psi}$ will have some finite stellar rank. In any case, by property~\ref{enum:vii}, for any $\epsilon>0$ the pure state $\ket{\bm\psi}$ is at least $\epsilon$-close in trace distance to a state $\smash{|\tilde{\bm\psi}\rangle}$ of finite stellar rank denoted $r^\star_\text{in}$. The trace distance is non-increasing under partial trace, so $\hat\rho$ is $\epsilon$-close in trace distance to $\smash{\tr_{\mathcal E}(|\tilde{\bm\psi}\rangle\!\langle\tilde{\bm\psi}|)}$. We also denote by $s$ the support size of the core state associated to the state $\smash{|\tilde{\bm\psi}\rangle}$.

By property~\ref{enum:vii}, the state $\smash{|\tilde{\bm\psi}\rangle}$ can be obtained by an optimization over ${\cal O}(m^2)$ parameters for any given rank. More precisely, the minimal achievable trace distance with a target $(2m)$-mode state $\ket{\bm\psi}$ using states of stellar rank less or equal to $k$ is given by~\cite[Theorem 2]{chabaud2020certification}:
\begin{equation}\label{eq:opti}
    \begin{aligned}
        \inf_{r^\star(\hat\sigma)\le k}D(\hat\sigma,\bm\psi)&=\sqrt{1-\sup_{\hat G\in\mathcal G}\tr\left[\Pi_k\hat G\ket{\bm\psi}\!\bra{\bm\psi}\hat G^\dag\right]}\\
        &=\sqrt{1-\sup_{\hat G\in\mathcal G}\sum_{|\bm n|\le k}\left|\bra{\bm n}\hat G\ket{\bm\psi}\right|^2},
    \end{aligned}
\end{equation}
where $\Pi_k=\sum_{|\bm n|\le k}\ket{\bm n}\!\bra{\bm n}$ and where the supremum on the right hand side is over a subset of $(2m)$-mode Gaussian unitary operations which can be parametrized using $(2m)^2-8m+1={\cal O}(m^2)$ real parameters. Importantly, the number of parameters to optimize is independent of the stellar rank of the approximation. Moreover, assuming the optimization yields a Gaussian unitary $\hat G_0$, then the corresponding approximating state is given by
\begin{equation}
    \ket{\tilde{\bm\psi}}=\hat G_0^\dag\left(\frac{\Pi_k\hat G_0\ket{\bm\psi}}{\left\|\Pi_k\hat G_0\ket{\bm\psi}\right\|}\right).
\end{equation}
The optimization is only necessary when the state $\ket{\bm\psi}$ has infinite stellar rank. While the function to optimize in Eq.~(\ref{eq:opti}) may become hard to compute for a large number of modes, in most cases it may be drastically simplified, e.g.\ if the infinite stellar rank state $\ket{\bm\psi}$ is a tensor product of single-mode states of infinite stellar rank (such as GKP states or cat states) followed by a multimode Gaussian unitary. Furthermore, it is not necessary to obtain an optimal solution of the optimization in the present case, but rather any solution which ensures $\epsilon$-closeness. Hence, we assume in what follows that the normal form of an approximating state as in Eq.~(\ref{eq:opti}) is available, with stellar rank $r^\star_\text{in}$ and core state support size $s$.

Similarly, we assume that each tensor product $\bigotimes_{k=1}^m\ket{y_k}\!\bra{y_k}$ has finite stellar rank $\sum_{k=1}^mr^\star(y_k)$, since a finite stellar rank approximation can be obtained efficiently. For all $k\in\{1,\dots,m\}$, we write $r^\star_k=\sup_{y_k\in\mathcal Y_k}r^\star(y_k)$. Finally, we write $r_\epsilon:=r^\star_\text{in}+\sum_{k=1}^mr^\star_k$, where the index indicates that this quantity depends on the choice of $\epsilon$, with 
\begin{equation}
    \begin{aligned}
        \lim_{\epsilon\rightarrow0}r_\epsilon&=r^\star(\bm\psi)+\sum_{k=1}^mr^\star_k\\
        &\ge r^\star(\hat\rho)+\sum_{k=1}^mr^\star_k\in\mathbb N\cup\{+\infty\},
    \end{aligned}
\end{equation}
where we used property~\ref{enum:vii} in the first line and property~\ref{enum:iii} in the second line (specifically the fact that the stellar rank is non-increasing under partial trace). We drop the $\epsilon$ index for brevity in what follows.

With these notations introduced, we now state our main result:

\begin{theorem}\label{th:strong}
    The measurement $\left\{\bigotimes_{k=1}^m\ket{y_k}\!\bra{y_k}\right\}_{ y\in\mathcal Y}$ of the state $\hat\rho$ can be $(2\epsilon)$-approximately strongly simulated classically in time ${\cal O}(s^2r^32^r+\text{poly }m)$.
\end{theorem}

\begin{proof}
By the preceding discussion, the state $\hat\rho=\tr_{\mathcal E}\left[\ket{\bm\psi}\!\bra{\bm\psi}\right]$ is $\epsilon$-close in trace distance to a state
\begin{equation}\label{eq:proofth1_1}
    \hat\sigma=\tr_{\mathcal E}\left[\ket{\tilde{\bm\psi}}\!\bra{\tilde{\bm\psi}}\right],
\end{equation}
where
\begin{equation}\label{eq:proofth1_2}
    |\tilde{\bm\psi}\rangle=\hat G_0^\dag\left(\frac{\Pi_{r^\star_\text{in}}\hat G_0\ket{\bm\psi}}{\left\|\Pi_{r^\star_\text{in}}\hat G_0\ket{\bm\psi}\right\|}\right)
\end{equation}
has finite stellar rank $r^\star_\text{in}$.
By the operational property of the trace distance and the triangle inequality, it thus suffices to show that the measurement ${\left\{\bigotimes_{k=1}^m\ket{y_k}\!\bra{y_k}\right\}_{ y\in\mathcal Y}}$ of the state $\hat\sigma$ can be $\epsilon$-approximately strongly simulated classically in time ${\cal O}(s^2r^32^r+\text{poly }m)$, where $r=r^\star_\text{in}+\sum_{k=1}^mr^\star_k$.
To do so requires two separate steps: 

\begin{itemize}
    \item Firstly, using the derivation from section~\ref{app:marginals}, we relate the corresponding marginal probabilities to marginal probabilities of Gaussian circuits with non-Gaussian input core states. 
    \item Secondly, we make use of the strong simulation algorithm of Gaussian circuits with non-Gaussian input core states from~\cite[Theorem~2]{chabaud2020classical}, which we recall below:
\end{itemize}

\begin{theorem}[\cite{chabaud2020classical}]\label{th:corestrong}
Let $m\in\mathbb N^*$ and let $\ket C$ be an $m$-mode core state of support size $s$ and stellar rank $n$. Then, Gaussian circuits over $m$ modes with input $\ket C$ and double homodyne detection can be strongly simulated classically in time ${\cal O}(s^2n^32^n +\text{poly }m)$.
\end{theorem}

\noindent Let $k\in\{1,\dots,m\}$. Being of finite stellar rank, the single-mode state $\ket{y_k}$ can be efficiently decomposed as
\begin{equation}
    \ket{y_k}=\frac1{\sqrt{{\cal N}_k}} \left[\prod_{j=1}^{r^\star(y_k)}\hat D(\beta_{k;j})\hat a^{\dag}_k\hat D^\dag(\beta_{k;j})\right]\hat S_k\ket{\alpha_k},
\end{equation}
where $\hat S_k$ is a squeezing operation, where $\alpha_k,\beta_{k;j}\in\mathbb C$ and where ${\cal N}_k$ is a normalisation factor.
By section~\ref{app:marginals}, any $p$-marginal probability (say without loss of generality the first $p$ modes) of the measurement ${\left\{\bigotimes_{k=1}^m\ket{y_k}\!\bra{y_k}\right\}_{ y\in\mathcal Y}}$ of the state $\hat\sigma$ approximated to $\left(\text{exp}(-\text{poly }m)\right)$-additive precision by a $p$-marginal probability of an associated coherent state sampler, which takes as input the state
\begin{equation}
    \begin{aligned}
        \hat\sigma^{(p)}_\text{total}:=&(\hat{\bm S}^{\dag} \otimes \hat{\mathds{1}}_{\rm aux})\hat {\cal U}^{\dag} \left(\hat \sigma \otimes\ket{1}\!\bra{1}^{\otimes\sum_{k=1}^p r^\star(y_k)}\right) \hat {\cal U} (\hat{\bm S}\otimes\hat{\mathds{1}}_{\rm aux})\\
        =&\tr_{\mathcal E}\left[(\hat{\bm S}^{\dag}\otimes \hat{\mathds{1}}_{\rm aux})\hat {\cal U}^{\dag}\left(\ket{\tilde{\bm\psi}}\!\bra{\tilde{\bm\psi}}\otimes\ket{1}\!\bra{1}^{\otimes\sum_{k=1}^p r^\star(y_k)}\right) \hat {\cal U} (\hat{\bm S}\otimes\hat{\mathds{1}}_{\rm aux})\right]\\
        =&\frac1{\left\|\Pi_{r^\star_\text{in}}\hat G_0\ket{\bm\psi}\right\|^2}\tr_{\mathcal E}\left[(\hat{\bm S}^{\dag}\otimes \hat{\mathds{1}}_{\rm aux})\hat {\cal U}^{\dag}(\hat G_0^\dag\otimes\hat{\mathds{1}}_{\rm aux}) \left(\Pi_{r^\star_\text{in}}\hat G_0\ket{\bm\psi}\!\bra{\bm\psi}\hat G_0^\dag\Pi_{r^\star_\text{in}}\otimes\ket{1}\!\bra{1}^{\otimes\sum_{k=1}^p r^\star(y_k)}\right)(\hat G_0\otimes\hat{\mathds{1}}_{\rm aux})\hat{\cal U}(\hat{\bm S}\otimes\hat{\mathds{1}}_{\rm aux})\right]\\
        =&\frac1{\left\|\Pi_{r^\star_\text{in}}\hat G_0\ket{\bm\psi}\right\|^2}\tr_{\mathcal E}\left[\hat G\left(\Pi_{r^\star_\text{in}}\hat G_0\ket{\bm\psi}\!\bra{\bm\psi}\hat G_0^\dag\Pi_{r^\star_\text{in}}\otimes\ket{1}\!\bra{1}^{\otimes\sum_{k=1}^p r^\star(y_k)}\right)\hat G^\dag\right],
    \end{aligned}
\end{equation}
where $\hat {\cal U}$ is given by $\hat {\cal U} := \bigotimes_{k=1}^p\prod_{n=1}^{r^\star(\tilde y_k)} \hat D(\beta_{k;j}) \hat U^{\dag}(\xi_{k;j}) \hat D^{\dag}(\beta_{k;j})$ with $\xi_{k;j}=\mathcal O\left(\text{exp}(-\text{poly }m)\right)$, where we used Eq.~(\ref{eq:marginaltotalstate}) in the first line, Eq.~(\ref{eq:proofth1_1}) in the second line, Eq.~(\ref{eq:proofth1_2}) in the third line, and where we defined the Gaussian unitary $\hat G:=(\hat{\bm S}^{\dag}\otimes \hat{\mathds{1}}_{\rm aux})\hat {\cal U}^{\dag}(\hat G_0^\dag\otimes\hat{\mathds{1}}_{\rm aux})$ in the last line.
By Eq.~(\ref{eq:marginalreduction}), the corresponding coherent state sampler marginal probability is given by
\begin{equation}
    \begin{aligned}
        \tilde P_{\bm\xi}(y_1,\dots, y_p \lvert \hat \sigma)&=\frac1{\cal N}\prod_{k=1}^p\prod_{j=1}^{r^\star(y_k)}\left(\frac{\cosh^2(\xi_{k;j})}{\sinh (\xi_{k;j})}\right)^2\tr\left[ \hat \sigma_{\rm total}^{(p)} \left(\bigotimes_{k=1}^p\ket{\alpha_k}\!\bra{\alpha_k}\otimes\bigotimes_{k=p+1}^m\hat{\mathds1} \otimes\ket0\!\bra0^{\otimes\sum_{k=1}^p r^\star(y_k)}\right) \right]\\
        &\propto\tr\left[\left(\Pi_{r^\star_\text{in}}\hat G_0\ket{\bm\psi}\!\bra{\bm\psi}\hat G_0^\dag\Pi_{r^\star_\text{in}}\otimes\ket{1}\!\bra{1}^{\otimes\sum_{k=1}^p r^\star(y_k)}\right)\hat G^\dag \left(\bigotimes_{k=1}^p\ket{\alpha_k}\!\bra{\alpha_k}\bigotimes_{j=1}^{m-p}\hat{\mathds1}\otimes\hat{\mathds1}_{\mathcal E}\otimes\ket{0}\!\bra{0}^{\otimes\sum_{k=1}^p r^\star(y_k)}\right)\hat G\right].
    \end{aligned}
\end{equation}
Now this expression corresponds to the marginal probability density of a Gaussian circuit $\hat G$ with double homodyne detection and non-Gaussian input core state
\begin{equation}
    \Pi_{r^\star_\text{in}}\hat G_0\ket{\bm\psi}\!\bra{\bm\psi}\hat G_0^\dag\Pi_{r^\star_\text{in}}\otimes\ket{1}\!\bra{1}^{\otimes\sum_{k=1}^p r^\star(y_k)}.
\end{equation}
This core state has the same support size as the state $\smash{|\tilde{\bm\psi}\rangle}$, which we denoted as $s$. Moreover, by property~\ref{enum:iii} the stellar rank is additive for pure states so this core state has stellar rank equal to $r_p:=r^\star_\text{in}+\sum_{k=1}^p r^\star(y_k)$.
We may now apply the strong simulation algorithm of Gaussian circuits with non-Gaussian input core states from~\cite{chabaud2020classical}: by Theorem~\ref{th:corestrong}, this marginal probability density can be evaluated exactly in time ${\cal O}(s^2r_p^32^{r_p} +\text{poly }m)$.
The same reasoning holds for all $p\in\{0,\dots,m\}$, so that all marginal probabilities of the associated coherent state sampler can be evaluated in time at most ${\cal O}(s^2r_m^32^{r_m} +\text{poly }m)$, where $r_m=r^\star_\text{in}+\sum_{k=1}^m r^\star(y_k)\le r$.

As a result, all probabilities and marginal probabilities of the measurement ${\left\{\bigotimes_{k=1}^m\ket{y_k}\!\bra{y_k}\right\}_{ y\in\mathcal Y}}$ of the state $\hat\sigma$ can be approximated to $(\text{exp}(-\text{poly }m))$-additive precision in time ${\cal O}(s^2r^32^r +\text{poly }m)$. By summing over all possible outcomes, this implies that this computation can be strongly simulated up to total variation distance $|\mathcal Y|\text{exp}(-\text{poly }m)$. Since $\mathcal Y$ has exponential size, this total variation distance can be brought to an arbitrarily small value, and smaller than $\epsilon$ in particular, with only a polynomial computational overhead in $m$. Hence, the measurement ${\left\{\bigotimes_{k=1}^m\ket{y_k}\!\bra{y_k}\right\}_{ y\in\mathcal Y}}$ of the state $\hat\sigma$ can be $\epsilon$-approximately strongly simulated in time ${\cal O}(s^2r^32^r +\text{poly }m)$, which completes the proof.

\end{proof}

\noindent Note that the time complexity in Theorem~\ref{th:strong} comes from the worst-case complexity of Theorem~\ref{th:corestrong}, based on the fastest known classical algorithm for computing the hafnian~\cite{bjorklund2019faster}, but this time complexity can be reduced for specific bosonic computations. For instance, for a Boson Sampling computation~\cite{10.1145/1993636.1993682} with $n$ input single photons and a unitary interferometer over $m$ modes with $m\times m$ unitary matrix $U$, our algorithm in Theorem~\ref{th:strong} computes hafnians of matrices of size at most $(4n)\times(4n)$. On the other hand, in the limit where the two-mode squeezing parameter $r$ goes to $0$, these hafnians reduce to permanents of submatrices of $U$ of size at most $n\times n$, which can be computed more efficiently.

\section{Efficient sampling from the \texorpdfstring{$Q$}{}-function of separable states}
\label{app:Qsamp}

In this section, we discuss the assumption from the main text that the $Q$-function of separable states can be efficiently sampled from classically.
It should be stressed that, even though a state's $Q$-function is a well-defined probability density, sampling from a highly multivariate probability density is not necessarily an easy computational task. It is therefore useful to set some basic assumptions on the sampling protocols that can be implemented efficiently. 

\medskip

From the theory of holomorphic functions, we know that $Q$-functions for single-mode states with a finite stellar rank are well-behaved \cite{chabaud2020stellar}. Furthermore, any single-mode state can be arbitrarily well approximated by a state with a finite stellar rank. This means that, for any sample size we want to generate from a certain state $\hat \sigma$, we can always find a state with finite stellar rank state that is sufficiently close to $\hat \sigma$ to have produced a statistically indistinguishable sample. Therefore, we assume that we can efficiently sample from every single-mode $Q$-function. This means that for a state $\hat \sigma = \hat \sigma_1 \otimes \dots \otimes \hat \sigma_N$ which is factorised in the measurement basis, we find that the corresponding probability density function can be written as $P(\vec \alpha \lvert \hat \sigma) = \prod_{k=1}^N P(\alpha_k \lvert \hat \sigma_k)$, and thus we can sample from this probability density by sampling each $\alpha_k$ individually from $P(\alpha_k \lvert \hat \sigma_k)$. 

However, when the state is not pure, and separable rather than factorisable, the situation is more subtle. In this case, we find that the state can be written as $\hat \sigma =\int_{\Lambda}d\lambda\, P(\lambda) \hat \sigma^{\lambda}_1 \otimes\dots \otimes \hat \sigma^{\lambda}_N  $, where $\lambda\in \Lambda$ labels the states in the mixture. The output probability density function can then be written as $\smash{P(\vec \alpha \lvert \hat \sigma) = \int_{\Lambda} d\lambda \, P(\lambda) \prod_{k=1}^N P(\alpha_k \lvert \hat \sigma^{\lambda}_k)}$. For any value of $\lambda$, we can still efficiently sample from the distributions $P(\alpha_k \lvert \hat \sigma^{\lambda}_k)$. Thus sampling from $P(\alpha_1, \dots, \alpha_N \lvert \hat \sigma)$ can still be efficiently done by first selecting a label $\lambda$ from the probability distribution $P(\lambda)$ and than sampling the $\alpha_k$ from $P(\alpha_k \lvert \hat \sigma^{\lambda}_k)$. 

Note that this procedure relies on two assumptions: first, we assume that the decomposition $\hat \sigma =\int_{\Lambda}d\lambda\, P(\lambda) \hat \sigma^{\lambda}_1 \otimes\dots \otimes \sigma^{\lambda}_N  $ is known; second, we assume that we can efficiently sample from $P(\lambda)$. In many sampling setups, these assumptions are not unreasonable since sampling experiments aim at using pure quantum resources. In quantum optics experiments in particular, the experimental imperfection that lead to mixed states can often be modeled with good accuracy such that the first assumption is reasonable. The second assumption can be grounded in the fact that decoherence in optical experiments is typically a Gaussian process, such that the labels $\lambda$ are effectively distributed according to a multivariate Gaussian.

\end{document}